\documentclass{CSML}
\pdfoutput=1

\usepackage{lastpage}
\newcommand{\dOi}{14(2:2)2018}
\lmcsheading{}{1--\pageref{LastPage}}{}{}%
{Feb.~17, 2017}{Apr.~10, 2018}{}

\usepackage{hyperref}
\hypersetup{hidelinks}

\usepackage{verbatim}
\usepackage{graphicx}
\usepackage{amssymb}
\usepackage{amsmath}
\usepackage[utf8]{inputenc}

\usepackage{amsthm}

\usepackage{xspace}

\usepackage{scalerel}

\usepackage[all]{xy}

\usepackage{float}

\usepackage{tikz}
\usetikzlibrary{shapes.geometric,arrows,calc,decorations.pathmorphing}

\usepackage[nodisplayskipstretch]{setspace}
\usepackage{amsmath}
\usepackage{enumitem}

\usepackage{esvect}

\makeatletter
\def\imod#1{\allowbreak\mkern10mu({\operator@font mod}\,\,#1)}
\makeatother

\newcommand\node[1]{*+[o]{#1}}

\newcommand\addLabelUL[1]{\ar@{}[]+UR|(1){~\makebox[0pt][l]{$\mathbf{#1}$}}}

\newcommand\addLabelUR[1]{\ar@{}[]+UR|(1){~\makebox[0pt][l]{$\mathbf{#1}$}}}

\newcommand\addLabelDL[1]{\ar@{}[]+UR|(1){~\makebox[0pt][l]{$\mathbf{#1}$}}}

\newcommand\addLabelDR[1]{\ar@{}[]+UR|(1){~\makebox[0pt][l]{$\mathbf{#1}$}}}

\newcommand\addDMD[2]{
	\ar@{-}[]+<#1pt,0pt>;[]+<0pt,#2pt>
	\ar@{-}[]+<0pt,#2pt>;[]+<-#1pt,0pt>
	\ar@{-}[]+<-#1pt,0pt>;[]+<0pt,-#2pt>
	\ar@{-}[]+<0pt,-#2pt>;[]+<#1pt,0pt>
}

\newcommand{\BM}{\mathbf{BM}}

\newcommand{\MSO}{\textnormal{MSO}}
\newcommand{\FO}{\textnormal{FO}}

\colorlet{ShapeLight}{gray!60}
\colorlet{ShapeDark}{gray!70!black!70}

\usetikzlibrary{shapes.geometric,arrows,calc,decorations.pathmorphing}

\newcommand{\shape}[1]{
\fill[#1]
	(0.2,1.2) --
	(0.7,3) --
	(1, 3) .. controls ++(0.5,-0.3) and ++(-0.5,-0.3) .. (2,3) --
	(3.5, 3) .. controls ++(0.5,-0.3) and ++(-0.5,-0.3) .. (5,3) --
	(5.5, 3) --
	(5.7, 1.5) --
	(5.7, 0.4) --
	(5.3,0) --
	(5, 0) .. controls ++(-0.5, +0.6) and ++(0.5,0.6) .. (4,0.1) --
	(4, 0.1) .. controls ++(-0.5, +0.6) and ++(0.5,0.6) .. (2.5,0) --
	(0.7,0) --
	(0.2,0.5);}

\newcommand{\boldclass}[3]{\ensuremath{\mathbf{#1}^{#2}_{#3}}}

\newcommand{\asigma}[1]{\boldclass{\Sigma}{1}{#1}}

\newcommand{\adelta}[1]{\boldclass{\Delta}{1}{#1}}

\newcommand{\U}{\mathcal{U}}

\newcommand{\qU}{\ensuremath{\mathsf U}\xspace}
\newcommand{\qM}{\ensuremath{\forall^{\ast}}\xspace}
\newcommand{\qN}{\ensuremath{\forall^{{=}1}}\xspace}
\newcommand{\qMpi}{\ensuremath{\forall^{\ast}_{\pi}}\xspace}
\newcommand{\qNpi}{\ensuremath{\forall^{{=}1}_{\pi}}\xspace}

\newcommand{\SoneS}{\ensuremath{S1S}\xspace}
\newcommand{\SoneSN}{\ensuremath{S1S+\qN}\xspace}
\newcommand{\SoneSM}{\ensuremath{S1S+\qM}\xspace}
\newcommand{\SoneSU}{\ensuremath{S1S+\qU}\xspace}

\newcommand{\StwoS}{\ensuremath{S2S}\xspace}
\newcommand{\StwoSN}{\ensuremath{S2S+\forall^{=1}}\xspace}
\newcommand{\StwoSM}{\ensuremath{S2S+\forall^{\ast}}\xspace}

\newcommand{\w}{\omega}

\newcommand{\Gdelta}{G_{\delta}}
\newcommand{\Fsigma}{F_{\sigma}}

\newcommand{\eqdef}{\stackrel{\text{def}}=}

\newcommand{\fun}[3]{\ensuremath{#1\colon #2 \to #3}}
\newcommand{\parfun}[3]{\ensuremath{#1\colon #2 \rightharpoonup #3}}

\newcommand{\Inf}{\mathrm{inf}\xspace}

\newcommand{\lang}{\mathrm{L}\xspace}

\newcommand{\init}{\mathrm{I}}

\newcommand{\xOper}{\mathbin{\chi}}
\DeclareMathOperator*{\XOper}{\scalerel*{\chi}{\sum}}

\newcommand{\yOper}{\mathbin{\psi}}
\DeclareMathOperator*{\YOper}{\scalerel*{\Psi}{\sum}}

\newcommand{\dom}{\mathrm{dom}\xspace}

\newcommand{\dL}{\mathtt{\scriptstyle L}}
\newcommand{\dR}{\mathtt{\scriptstyle R}}

\newcommand{\defPlayer}[1]{\ensuremath{\boldsymbol{#1}}\xspace}

\newcommand{\eve}{\defPlayer{\exists}}
\newcommand{\adam}{\defPlayer{\forall}}

\newcommand{\succL}{\textnormal{succ}_L}
\newcommand{\succR}{\textnormal{succ}_R}

\newcommand{\powerset}{\mathcal{P}}

\usetikzlibrary{decorations}
\usetikzlibrary{decorations.pathmorphing}
\usetikzlibrary{decorations.markings}
\usetikzlibrary{decorations.pathreplacing}

\newcommand{\tikzEvalInt}[2]{\pgfmathparse{int(#2)}{\global\edef#1{\pgfmathresult}}}

\tikzstyle{flow}  = [inner sep=0cm, node distance=0cm and 0cm]
\tikzstyle{toL} = [anchor=mid east, flow]
\tikzstyle{toR} = [anchor=mid west, flow]
\tikzstyle{toC} = [align=center, anchor=mid, flow]

\tikzstyle{ubrace} = [draw, thick, decoration={brace, mirror, raise=0.0cm}, decorate, every node/.style={anchor=north, yshift=-0.1cm}]

\newcommand{\trees}[1]{{\mathcal T}_{#1}}
\newcommand{\mathcalsym}[1]{\ensuremath{\mathcal{#1}}\xspace}
\newcommand{\ignore}[1]{}

\newcommand{\Aa}{{\mathcal A}}
\newcommand{\Bb}{\mathcal{B}}

\newcommand{\Ff}{\mathcalsym{F}}

\newcommand{\SuccL}{\texttt{Succ}_L}

\newcommand*{\diameter}{0.2cm}
\newcommand*{\smalldiameter}{0.15cm}
\newcommand{\la}{\langle}
\newcommand{\ra}{\rangle}

\usetikzlibrary{matrix}

\tikzset{ 
    table/.style={
        matrix of nodes,
        row sep=-\pgflinewidth,
        column sep=-\pgflinewidth,
        nodes={
            rectangle,
            draw=black,
            align=center
        },
        minimum height=0.7em,
        text depth=0.5ex,
        text height=2ex,
        nodes in empty cells,
        column 1/.style={
            nodes={text width=5ex,font=\bfseries}
        },
        row 1/.style={
            nodes={
                fill=gray!20,
                text=black,
                font=\bfseries
            }
        }
    }
}

\usepackage{stmaryrd} 
\usepackage{amssymb}

\begin{document}

\begin{abstract}
We investigate the extension of Monadic Second Order logic, interpreted over infinite words and trees, with generalized ``for almost all'' quantifiers interpreted using the notions of Baire category and Lebesgue measure. 
\end{abstract}

\author[M.~Mio]{Matteo Mio}	
\address{CNRS and \'{E}cole Normale Sup\'{e}rieure de Lyon}	
\email{matteo.mio@ens-lyon.fr}  
\thanks{All three authors were supported by the Polish National Science Centre grant no. 2014-13/B/ST6/03595. The work of the second author has also been supported by the project ANR-16-CE25-0011 REPAS.}	

\author[M.~Skrzypczak]{Micha{\l} Skrzypczak}	
\address{Institute of Informatics, University of Warsaw}	
\email{m.skrzypczak@mimuw.edu.pl}

\author[H.~Michalewski]{Henryk Michalewski}	
\address{Institute of Mathematics, Polish Academy of Sciences}	
\email{h.michalewski@mimuw.edu.pl}  

\title{Monadic Second Order Logic with Measure and Category Quantifiers}

\maketitle              

\section{Introduction}
Monadic Second Order logic (\MSO) is an extension of first order logic (\FO) allowing quantification over subsets of the domain ($\forall X.\,\phi$) and including a membership predicate between elements of the domain and its subsets ($x \in  X$). Two fundamental results about \MSO, due to B{\"u}chi and Rabin respectively, state that the \MSO\  theories of $(\mathbb{N},<)$ and of the \emph{full binary tree} (see Section~\ref{section_mso_1} for precise definitions) are decidable.

A different kind of extension of first order logic which has been studied in the literature, starting from the seminal works of Mostowski~\cite{Mostowski1957}, is by means of \emph{generalized quantifiers}. A typical example is the ``there exists infinitely many'' quantifier:
\[
\mathcal{M}\models \exists^\infty x.\, \phi \ \Longleftrightarrow \text{ the set $\{m\in M \mid \phi(m)\}$ is infinite.} 
\]
The formula $\neg (\exists^\infty x.\, (x=x))$ holds only on finite models and is therefore not expressible in first order logic. Hence, in general, $\FO \subsetneq  \FO +\exists^\infty$. Note however that on some fixed structures the logic $\FO$ and $\FO+\exists^\infty$ may have the same expressive power. This is the case, for example, for the structure $(\mathbb{N},{<})$.

It is natural to try to merge the two approaches and consider extensions of \MSO\ with generalized quantifiers. In this context, one can consider generalized first order quantifiers as well as generalized second order quantifiers. For example, the second order generalized quantifier ($\exists^{\infty}X.\phi$) has the following meaning:

\[
\mathcal{M}\models \exists^\infty X.\, \phi \ \Longleftrightarrow \text{the set $\{A\subseteq M \mid \phi(A)\}$ is infinite.} 
\]

The topic of  $\MSO$ with generalized second order quantifiers is not new. In particular, Boja{\'n}czyk
in 2004 (\cite{bojanczyk2004}) has initiated an investigation of the \emph{unbounding quantifier} ($\qU X.\,\phi$) 
\[
\mathcal{M}\models \qU X.\, \phi \ \Longleftrightarrow  \forall n\in\mathbb{N}. \, \exists A\subseteq M. \,  \big(\text{$A$ is finite, $|A|\geq n$, $\phi(A)$}\big).
\]
It was shown recently that the $\MSO+\qU$ theory of $(\mathbb{N},{<})$ is undecidable~\cite{bpt2016}.

B\'{a}r\'{a}ny, Kaiser and Rabinovich have investigated in~\cite{KRB2008} the second order quantifier $\exists^\infty$ mentioned above and also the cardinality quantifier $\exists^{\mathfrak{c}}$ defined as:
\[
\mathcal{M}\models \exists^{\mathfrak{c}}X.\, \phi \Longleftrightarrow \text{the set $\{A\subseteq M \mid \phi(A) \}$ has cardinality $\mathfrak{c}=2^{\aleph_0}$.} 
\]
Interestingly it was shown that, on the structures $(\mathbb{N},{<})$ and the full binary tree, the quantifiers $\exists^\infty$ and $\exists^{\mathfrak{c}}$ can be eliminated. That is, to each formula of $\MSO +\exists^{\mathfrak{c}}+\exists^\infty$ corresponds (effectively) a semantically equivalent formula of $\MSO$.

In a series of manuscripts in the late 70's (see~\cite{steinhorn} for an overview) H.~Friedman investigated  two kinds of first-order generalized quantifiers called (Baire) \emph{Category quantifier} $(\qM)$ and (full probability) \emph{Measure quantifier} ($\qN$), respectively. 

The goal of this article is to collect recent results (from~\cite{MM2015},~\cite{MM2015b} and~\cite{MM2016LFCS}) regarding extensions of $\MSO$ with second order Category and Measure quantifiers. The meaning of such second order quantifiers is given as follows: 
\[
\mathcal{M}\models \qM X.\, \phi \Longleftrightarrow \text{the set $\{A\subseteq M \mid \phi(A) \}$ is comeager.}
\]
and
\[
\mathcal{M}\models \qN X.\, \phi \Longleftrightarrow \text{the set $\{A\subseteq M \mid \phi(A) \}$ has Lebesgue measure $1$.}
\]

We will exclusively be interested in the properties of $\MSO+\qM$ and $\MSO+\qN$ interpreted over the two fundamental structures $(\mathbb{N},{<})$ and the full binary tree. Therefore, the domain of $\mathcal{M}$ is always countable and the space of subsets $A\subseteq M$ is a copy of the Cantor space $2^\w$ where the notions of (co)meager sets and Lebesgue (i.e., coin-flipping) measure are well defined. See Section~\ref{sec:tech_back} for precise definitions. 

\paragraph*{\bf Results regarding $\SoneS$.} By denoting with $\SoneS$ the $\MSO$ theory of $(\mathbb{N},{<})$ we prove (Theorem~\ref{thm:sonesMelim}) that the Category quantifier ($\qM$) can be (effectively) eliminated. Therefore $\SoneSM= \SoneS$ and its theory is decidable. On the other hand, the addition of the Measure quantifier ($\qN$) leads to an undecidable theory (Theorem~\ref{undec:s1s:measure}). This was proved in~\cite{MM2015} using recent results from the theory of probabilistic automata. In this article we give a different proof of undecidability of $\SoneSN$ based on an interpretation of Boja{\'n}czyk's undecidable theory $\SoneSU$.

\paragraph*{\bf Results regarding $\StwoS$.} By denoting with $\StwoS$ the $\MSO$ theory of the full binary tree, the theory of $\SoneSN$ can be interpreted within $\StwoSN$. As a consequence $\StwoSN$ has an undecidable theory (Theorem~\ref{thm:msoqn:undecidable}).

In~\cite{MM2015} it was claimed that the theory $\StwoSM$ is effectively equivalent to $\StwoS$. The proof relied on an automata-based quantifier elimination procedure. Unfortunately, the proof contains a crucial mistake which we do not know how to fix. The mistake is caused by a wrong citation in Proposition~1: although the automaton $\Aa^{\exists}$ recognizes the language $\exists X.\,\phi(X,\vv{Y})$, this may not be the case for the automaton $\Aa^{\forall}$ and the language $\forall X.\,\phi(X,\vv{Y})$. Because of that, the automaton $\Aa_\psi$ constructed in Definition~6 does not fairly simulate the Banach--Mazur game in the moves of the universal player \adam.

The proof technique adopted in~\cite{MM2015}, while faulty in the general case, leads to the following weaker result (Theorem~\ref{thm:game-meager}): if $\phi(X,\vec{Y})$ is a $\StwoS$ formula (with parameters $X$, $\vec{Y}$) definable by a \emph{game automaton}, then the quantifier $\qM X.\, \phi(X,\vec{Y})$ can be effectively eliminated. That is, there (effectively) exists a $\StwoS$ formula $\psi(\vec{Y})$ equivalent to $\qM X.\, \phi(X,\vec{Y})$. Comparing to the proof presented in~\cite{MM2015}, the proof in this paper is shorter and more precise. 


The question of whether the full logic $\StwoSM$ has a decidable theory remains open from this work. See the list of open problems in Section~\ref{sec:problems}.

\paragraph*{\bf Results regarding $\StwoS$ with generalized path quantifiers.} Another interesting way to extend $\StwoS$ with variants of Friedman's Category and Measure quantifiers is to restrict the quantification to range over infinite branches (paths) of the full binary tree. This leads to the definition of the path-Category ($\qMpi X.\, \phi$) and the path-Measure ($\qNpi X.\, \phi$) second-order quantifiers.

A path $\pi$ in the full binary tree is an infinite set of vertices which can be uniquely identified with an infinite set of direction $\pi\in \{\dL,\dR\}^\w$. By denoting with $\mathcal{P}$ the collection $\{\dL,\dR\}^\w$ we define:
\[
 \qMpi X.\, \phi \Longleftrightarrow \text{$\{\pi\in\mathcal{P} \mid \phi(\pi)\}$ is comeager as a subset of $\mathcal{P}$}
\]
and
\[
 \qNpi X.\, \phi \Longleftrightarrow \text{$\{\pi\in\mathcal{P} \mid \phi(\pi)\}$ has measure $1$ as a subset of $\mathcal{P}$}
\]


We show that the path-Category quantifier can be (effectively) eliminated so that $\StwoS+\qMpi\equiv \StwoS$. On the other hand, we show that $\StwoS+\qNpi \supsetneq \StwoS$ is a strict extension of $\StwoS$.  The question of whether $\StwoS+\qNpi$ has a decidable theory remains open from this work. See the list of open problems in Section~\ref{sec:problems}.


\paragraph*{\bf Related Work.} A significant amount of work has been devoted to the study of properties of regular sets of words and infinite trees with respect to Baire Category and probability. For example, a remarkable result of Staiger (\cite[Theorem 4]{staiger}) states that  a regular set $L\!\subseteq\!\Sigma^\omega$ of $\omega$-words is comeager 
if and only if $\mu(L)\!=\!1$, where $\mu$ is the probability (coin flipping) measure on $\omega$-words. Furthermore the probability $\mu(L)$ can be easily computed by an algorithm based on linear programming and it is always a rational number (see, e.g., Theorem 2 of~\cite{Chat2004}). Staiger's theorem was rediscovered by Varacca and V{\"o}lzer and has been used to develop a theory of fairness for reactive systems~\cite{Var2012}.

In another direction, Gr{\"a}edel have investigated determinacy, definability, and complexity issues of Banach--Mazur games played on graphs~\cite{Graedel2008}. Our results regarding $\SoneSM$ are closely related to the results of Gr{\"a}edel.

The study of Category and Measure properties of regular sets of infinite trees happens to be more complicated. Indeed, due to their high topological complexity, it is not even clear if such sets are Baire-measurable and Lebesgue measurable. These properties were established only recently (see~\cite{hjorth} for Baire measurability and~\cite{GMMS2014} for Lebesgue measurability). Furthermore, it has been shown in~\cite{MM2015b} that the equivalent of Staiger's theorem fails for infinite trees. That is, there exist comeager regular sets of trees having probability $0$.

The problem of computing the probability of a regular set of infinite trees is still open in the literature. An algorithm for computing the (generally irrational) probability of regular sets definable by game automata is presented in~\cite{MM2015b} but a solution for the general case is still unknown (and the problem could be undecidable). The result of Theorem~\ref{thm:game-meager} proved in this article imply that it is possible to determine the Baire category of regular sets definable by game automata. A solution to the general case is unknown.

All the results mentioned above contribute to the theory of regular sets of $\omega$-words ($\SoneS$ definable) and infinite trees ($\StwoS$ definable). The novelty of the approach followed in this research is instead to augment the base logic ($\SoneS$ or $\StwoS$) with generalized quantifiers capable of expressing properties related to Baire category and probability, and to study such logical extensions.

Interestingly, a different but related approach has been explored by Carayol, Haddad and Serre in~\cite{CHS2014} (see also~\cite{carayol_serre,carayol_counting} and the habilitation thesis~\cite{serre2015}).
Rather than extending the base logic, it is the base (Rabin) tree automata model which is modified. The usual acceptance condition ``there exists a run such that \emph{all} paths are accepting'' is replaced by the weaker ``there exists a run such that \emph{almost all} paths are accepting``. When \emph{almost all} is interpreted as ``for a set of measure $1$'', this results in the notion of qualitative automata of~\cite{CHS2014}. 
The connections between this automata-oriented and our logic-oriented approach is further discussed in Section~\ref{sec:s2s-measure-pi}.

\emph{Addendum}: after the submission of this article, we have been informed by M.~Boja{\'n}czyk that a proof of the decidability of the theory of $\mathrm{weak}\StwoS+\qNpi$, the fragment of $\StwoS+\qNpi$ where second-order quantification is restricted to finite sets, can be obtained from the results of~\cite{bojanczyk_thin_mso,bojanczyk_subzero}.

\section{Technical Background}
\label{sec:tech_back}

In this section we give the basic definitions regarding descriptive set theory, probability measures and tree automata that are needed in this work.

The set of natural numbers and their standard total order are denoted by the symbols $\w$ and ${<}$, respectively. Given sets $X$ and $Y$ we denote with $X^Y$ the space of functions $X\to Y$. We can view elements of $X^Y$ as $Y$-indexed sequences $\{x_i\}_{i\in Y}$ of elements of $X$.  
Given a finite set $\Sigma$, we refer to $\Sigma^\w$ as the collection of \emph{$\omega$-words over $\Sigma$}. The 
 collection of \emph{finite} sequences of elements in $\Sigma$ is denoted by $\Sigma^\ast$.  As usual we denote with $\epsilon$ the empty sequence and with $ww'$ the concatenation of $w,w'\in\Sigma^\ast$. The prefix order on sequences is denoted ${\preceq}$ i.e., $u\preceq w$ if $u$ is a prefix of $w$.
 
We identify finite sequences $\{\dL,\dR\}^\ast$ as vertices of the infinite full binary tree so that $\epsilon$ is the root and each vertex $w\in \{\dL,\dR\}^\ast$ has $w\dL$ and $w\dR$ as children. We denote by $\trees{\Sigma}$ the set $\Sigma^{\{\dL,\dR\}^\ast}$ of functions labelling each vertex with an element in $\Sigma$. Elements $t\in\trees{\Sigma}$ are called \emph{trees over $\Sigma$}.

\subsection{Games}
\label{sec:games}

In this paper we work with games of perfect information and infinite duration commonly referred to as Gale--Stewart games. In what follows we provide the basic definitions. We refer to~\cite[Chapter~20.A]{Kechris} for a detailed overview.

Gale--Stewart games are played between two players: Player~\eve and Player~\adam. A game is given by a set of configurations $V$; an initial configuration $v_0\in V$; a winning condition $W\subseteq V^\omega$; and a definition of a \emph{round}. A \emph{round} starts in a configuration $v_n$; then players perform a~fixed sequence of choices; and depending on these choices the round ends in a new configuration $v_{n+1}$. Each game specifies what are the exact choices available to the players and how the successive configuration $v_{n+1}$ depends on the decisions made.


A \emph{play} of such a game is an infinite sequence of configurations $v_0,v_1,\ldots$ where the configuration $v_{n+1}$ is a possible outcome of the $n$-th round that starts in a configuration $v_n$. A play is winning for Player~\eve if it satisfies the winning condition (i.e.~$(v_0,v_1,\ldots)\in W$); otherwise the play is winning for Player~\adam. A \emph{strategy} of a player is a function that assigns to each finite sequence of configurations $(v_0,v_1,\ldots,v_n)$ a way of resolving the choices of that player in the consecutive $n$-th round that starts in $v_n$. A play is \emph{consistent} with a strategy if it can be obtained following that strategy. A strategy of a player is \emph{winning} if all the plays consistent with that strategy are winning for the player. We say that a game is \emph{determined} if one of the players has a winning strategy.

\subsection{Topology}
\label{dst}
Our exposition of topological and set--theoretical notions follows~\cite{Kechris}.

The set $\{0,1\}^\w$ endowed with the product topology (where $\{0,1\}$ is given the discrete topology) is called the \emph{Cantor space}. The Cantor space is a Polish  space (i.e., a complete separable metric space) and is \emph{zero-dimensional} (i.e., it has a basis of \emph{clopen} sets). Given a finite set $\Sigma$ with at least two elements, the spaces $\Sigma^{\omega}$ and $\trees{\Sigma}$ are both homeomorphic to the Cantor space. 

A subset  $A\subseteq X$ of a topological space is \emph{nowhere dense} if the interior of its closure is the empty set, that is ($\texttt{int}(\texttt{cl}(A))= \emptyset$. A set $A \subseteq X$ is of \emph{(Baire) first category} (or \emph{meager}) if $A$ can be expressed as countable union of nowhere dense sets. A set $A \subseteq X$ which is not meager is \emph{of the second (Baire) category}. 
The complement of a meager set is called \emph{comeager}. 
Recall that a set $A\subseteq \{0,1\}^\omega$ is a \emph{$G_\delta$ set} if it can be expressed as a countable intersection of open sets. The following property of comeager sets will be useful. A subset $B\subseteq \{0,1\}^\w$ is comeager if and only if there exists a dense $G_\delta$ set $A\subseteq B$ which is comeager.

A set $B \subseteq X$ has the \emph{Baire property} if $X = U\triangle M$, for some open set $U \subseteq X$ and meager set $M \subseteq X$, where $\triangle$ is the operation of symmetric difference $X\triangle Y =  (X\cup Y)\setminus (X\cap Y)$.


We now describe Banach--Mazur game (see~\cite[Chapter~8.H]{Kechris} for a detailed overview) which characterizes Baire category. We specialize the game to zero--dimensional Polish spaces, such as $\Sigma^\w$ and $\trees{\Sigma}$, since we only deal with this class of spaces in this work. 


\begin{defi}[Banach--Mazur Game]
\label{bmgames} Let $X$ be a zero--dimensional Polish space. 
Fix a payoff set $A \subseteq  X$. The configurations of the game $\BM(X,A)$ are non-empty clopen sets $U\subseteq X$; the initial configuration is $X$. Consider a round starting in a configuration $U_n$. In such a round, first Player~\adam chooses a non-empty clopen set $U'_n\subseteq U_n$ and then Player~\eve chooses a non-empty clopen set $U_{n+1}\subseteq U'_{n}$. The consecutive configuration is $U_{n+1}$. An infinite play of that game is won by Player~\eve if $\displaystyle\bigcap_{n\in\omega} U_n\cap A\neq\emptyset$ and Player~$\adam$ wins otherwise.
\end{defi}

The standard way of defining that game is to represent it as a form of a Gale--Stewart game. 

\begin{center}
\begin{tabular}{l l l l l l l }
Player~$\adam$ & &  $U'_0$ &        &  $U'_1$ &    & $\dots$\\
Player~$\eve$ & &         &  $U_1$ &       &  $U_2$ & $\dots$ 
\end{tabular}
\end{center}

For a proof of the following result see, e.g.,~Theorem~8.33 in~\cite{Kechris}.

\begin{thm}
\label{bm_game_theorem}
Let $X$ be a zero--dimensional Polish space. If $A \subseteq X$ has the Baire property, then $\BM(X,A)$ is determined and Player~$\eve$ wins iff $A$ is comeager.
\end{thm}

If $X$ is the space of $\w$-words or infinite trees over a finite alphabet $\Sigma$, the above game can be simplified as in Definitions~\ref{def:bm-words} and~\ref{def:bm-trees} below.

\begin{defi}
\label{def:bm-words}
Fix a language $L\subseteq \Sigma^\w$ of $\w$-words over $\Sigma$. The \emph{Banach--Mazur game} on $L$ (denoted $\BM(L)$) is the following game with configurations of the form $(s,r)$ where $s\in\Sigma^\ast$ is a finite word and $r\in\{\eve,\adam\}$. The initial configuration $(s_0,r_0)$ is $(\epsilon,\adam)$. Consider a round $n=0,1,\ldots$ that starts in a configuration $(s_n,r_n)$. In this round, the player $r_n$ chooses a finite word $s_{n+1}$ that contains $s_n$ as a strict prefix and the round ends in the configuration $(s_{n+1}, \bar{r_n})$ where $\bar{r_n}$ is the opponent of $r_n$.

An infinite play of this game is winning for $\eve$ if $\alpha\in L$, where $\alpha=\bigcup_{n\in\w} s_n$ is the unique $\w$-word that has all $s_n$ as prefixes.
\end{defi}

Now we move to the tree variant of this game. A \emph{tree-prefix} is a partial function $\parfun{s}{\{\dL,\dR\}^\ast}{\Sigma}$ where the domain $\dom(s)$ is finite, prefix-closed, and each node of $\dom(s)$ has either zero (a leaf) or two (an internal node) children in $\dom(s)$. We say that $s'$ \emph{extends} $s$ if $s'\neq\emptyset$, $s\subseteq s'$, and each leaf of $s$ is an internal node of $s'$.

\begin{defi}
\label{def:bm-trees}
Fix a language $L\subseteq \trees{\Sigma}$ of trees over $\Sigma$. The \emph{Banach--Mazur game} on $L$ (denoted $\BM(L)$) is the following game with configurations of the form $(s,r)$ where $s$ is a tree-prefix and $r\in\{\eve,\adam\}$. The initial configuration $(s_0,r_0)$ is $(\emptyset,\adam)$. Consider a round $n=0,1,\ldots$ that starts in a configuration $(s_n,r_n)$. In this round, the player $r_n$ chooses a tree-prefix $s_{n+1}$ that extends $s_n$ and the round ends in the configuration $(s_{n+1}, \bar{r_n})$ where $\bar{r_n}$ is the opponent of $r_n$. An infinite play of this game is winning for Player~$\eve$ if $t\in L$, where  $t=\bigcup_{n\in\w} s_n$.
\end{defi}

\begin{thm}
A language $L$ of $\w$-words or trees over $\Sigma$ with the Baire property is comeager if and only if $\eve$ has a winning strategy in the Banach--Mazur game on $L$.
\end{thm}

\subsection{Probability Measures}

We summarize below the basic concepts related to Borel measures. For more details see, e.g,~\cite[Chapter~17]{Kechris}.

In what follows, let $X$ be a zero--dimensional Polish space. The smallest $\sigma$-algebra of subsets of $X$ containing all open sets is denoted by $\mathcal{B}$ and its elements are called \emph{Borel sets}. 
A \emph{Borel probability measure} on $X$ is a function $\fun{\mu}{\mathcal{B}}{[0,1]}$ such that: $\mu(\emptyset)=0$, $\mu(X)=1$ and, if $\{B_n\}_{n\in \w}$ is a sequence of disjoint Borel sets, $\mu\big(\bigcup_{n}B_n\big)= \sum_{n}\mu(B_n)$. 

A subset $A\subseteq X$ is a \emph{$\mu$-null set} if $A\subseteq B$ for some Borel set $B\subseteq X$ such that $\mu(B)=0$.

A subset $A\subseteq X$ is a \emph{$F_\sigma$ set} if it can be expressed as a countable union of closed sets. Every Borel measure $\mu$ on a $X$ is \emph{regular}: for every Borel set $B$ there exists an $F_\sigma$ set $A\subseteq B$ such that $\mu(A)=\mu(B)$.  

\paragraph*{\bf Lebesgue measure on $\w$-words.} 
For a fixed finite non-empty set $\Sigma=\{a_1,\dots, a_k\}$, we will be exclusively interested in one specific Borel measure $\mu_\w$ on  the space $\Sigma^\w$ which we refer to as \emph{Lebesgue measure} (on $\w$-words). This is the unique Borel measure satisfying the equality $\mu_w(B_{n=a})=\frac{1}{k}$ where $B_{n=a}=\{ (a_i)_{i\in \w} \mid a_n = a\}$, for each $n\in\mathbb{N}$ and $a\in \Sigma$. Intuitively, the Lebesgue measure on $\w$-words generates an infinite sequence $(a_0,a_1,\dots)$ by labelling the $n$-th entry with some letter $a\in \Sigma$ chosen randomly and uniformly. This is done independently for each position $n\in\w$.

\paragraph*{\bf Lebesgue measure on trees.} 

Similarly, for a finite non-empty set $\Sigma=\{a_1,\dots, a_k\}$, the \emph{Lebesgue measure} on trees $\trees{\Sigma}$ is the unique Borel measure, denoted by $\mu_t$, satisfying  $\mu_t(B_{v=a})=\frac{1}{k}$ where $B_{v=a}=\{ t\in \Sigma^{\{\dL,\dR\}^\ast} \mid t(v) = a\}$, for $v\in\{\dL,\dR\}^\ast$ and $a\in \Sigma$.  Once again, the Lebesgue measure on trees generates an infinite tree over $\Sigma$ by labelling each vertex $v$ with a randomly chosen letter $a\in \Sigma$. This is done independently for each vertex $v\in\{\dL,\dR\}^\ast$.

\subsection{Large Sections Projections}

Let $X$ and $Y$ be Polish spaces. We denote with $X\times Y$ the product space endowed with the product topology. Given a subset $A\subseteq X\times Y$ and an element $x\in X$ we denote with $A_x\subseteq Y$ the \emph{section of $A$ at $x$} defined as: $A_x \eqdef \{ y\in Y \mid  (x,y)\in A\}$.

A fundamental operation in descriptive set theory, corresponding to the logical operation of universal quantification, is that of taking full sections: given $A\subseteq X\times Y$ we denote with $\forall_{Y}A \subseteq X$ the set 
\[
\forall_Y A \eqdef \{ x \in X \mid A_x=Y\}
\]
In other words, $\forall_Y A$ is the collection of those $x\in X$ such that for all $y\in Y$ the pair $(x,y)$ is in $A$, see Figure~\ref{figure_full_section}. 

\begin{figure}[H]
\centering
\begin{tikzpicture}[scale=1.1]

\tikzstyle{axis} = [->, very thick, line cap = rect]
\tikzstyle{theSet} = [color=black, line width=4pt, line cap=round, shorten <=1pt, shorten >=1pt]

\shape{fill=ShapeLight}

\begin{scope}
\clip(0.7,-1) rectangle (1,4);
\shape{fill=ShapeDark}
\end{scope}

\draw[theSet] (0.7, 0) -- (1.0,0);

\begin{scope}
\clip(2,-1) rectangle (2.5,4);
\shape{fill=ShapeDark}
\end{scope}

\draw[theSet] (2, 0) -- (2.5,0);

\begin{scope}
\clip(5,-1) rectangle (5.3,4);
\shape{fill=ShapeDark}
\end{scope}

\draw[theSet] (5, 0) -- (5.3,0);

\draw[axis] (0,0)  -- ++(6,0) node(xline)[right] {$\vv{Y}$};
\draw[axis] (0,0) -- ++(0,3.1) node(yline)[above] {$X$};
\node at (5.8,3.2) {$\phi(x,\vv{y})$};
\node (phi) at (3,-1.2) {$\forall x.\,\phi(x,\vv{y})$};
\draw[->] (phi) -- (0.85,-0.2);
\draw[->] (phi) -- (2.25,-0.2);
\draw[->] (phi) -- (5.15,-0.2);

\end{tikzpicture}
\caption{The universal quantifier $\forall$. A formula $\forall x.\,\phi(x,\protect\vv{y})$ holds for parameters $\protect\vv{y}\in\protect\vv{Y}$ if for every choice of $x\in X$ we have $\phi(x,\protect\vv{y})$. The full sections selected by this quantifier are marked in gray.}
\label{figure_full_section}
\end{figure}
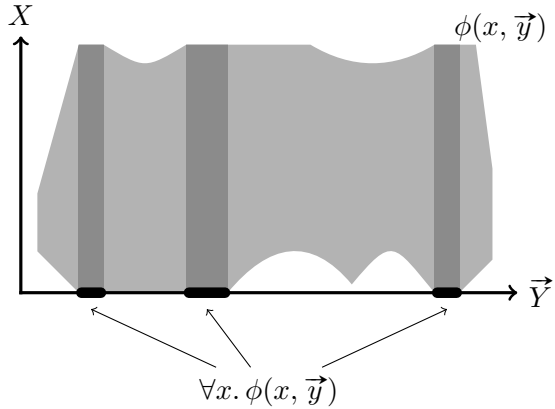  

\begin{defi}[Projective sets]
The family of projective sets is the smallest family containing Borel sets and closed under Boolean operations and full sections.
\end{defi}

The content of this article is based on other operations studied in descriptive set theory (see, e.g.,~\cite[\S~29.E]{Kechris}) based on the idea of taking \emph{large sections}, rather than full sections, see Figure~\ref{figure_large_section}. From a logical point of view, these operations corresponds to some kind of weakened ``for almost all'' quantification.

The two notions of ``largeness'' that are relevant in this work are that of being \emph{comeager} and having \emph{Lebesgue measure $1$}, see Figure~\ref{figure_large_section}. 

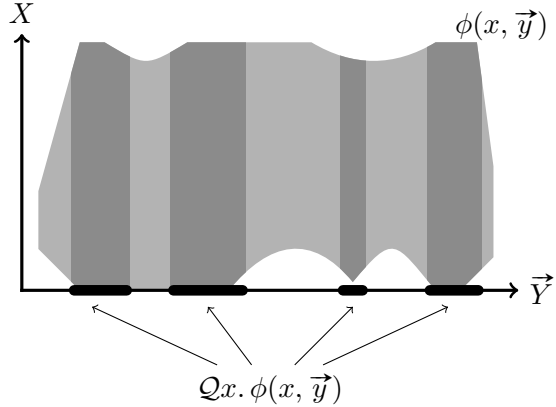
\begin{figure}[H]
\centering
\begin{tikzpicture}[scale=1.1]

\tikzstyle{axis} = [->, very thick, line cap = rect]
\tikzstyle{theSet} = [color=black, line width=4pt, line cap=round, shorten <=1pt, shorten >=1pt]

\shape{fill=ShapeLight}

\begin{scope}
\clip(0.6,-1) rectangle (1.3,4);
\shape{fill=ShapeDark}
\end{scope}

\draw[theSet] (0.6, 0) -- (1.3,0);

\begin{scope}
\clip(1.8,-1) rectangle (2.7,4);
\shape{fill=ShapeDark}
\end{scope}

\draw[theSet] (1.8, 0) -- (2.7,0);

\begin{scope}
\clip(3.85,-1) rectangle (4.15,4);
\shape{fill=ShapeDark}
\end{scope}

\draw[theSet] (3.85, 0) -- (4.15,0);

\begin{scope}
\clip(4.9,-1) rectangle (5.55,4);
\shape{fill=ShapeDark}
\end{scope}

\draw[theSet] (4.9, 0) -- (5.55,0);

\draw[axis] (0,0)  -- ++(6,0) node(xline)[right] {$\vv{Y}$};
\draw[axis] (0,0) -- ++(0,3.1) node(yline)[above] {$X$};
\node at (5.8,3.2) {$\phi(x,\vv{y})$};
\node (phi) at (3,-1.2) {$\mathcal{Q} x.\,\phi(x,\vv{y})$};
\draw[->] (phi) -- (0.85,-0.2);
\draw[->] (phi) -- (2.25,-0.2);
\draw[->] (phi) -- (4,-0.2);
\draw[->] (phi) -- (5.15,-0.2);

\end{tikzpicture}
\caption{A large section quantifier $\mathcal{Q}$, for instance $\mathcal{Q}$ can be $\forall$, $\qM$, $\qN$, etc. A formula $\mathcal{Q} x.\,\phi(x,\protect\vv{y})$ holds for parameters $\protect\vv{y}\in\protect\vv{Y}$ if there are in some sense many (but possibly not all) values $x\in X$ such that $\phi(x,\protect\vv{y})$ holds. The large sections selected by this quantifier are marked in gray.}
\label{figure_large_section}
\end{figure}  

Correspondingly, given $A\subseteq X\times Y$, we denote with $\qM_{Y} A \subseteq X$ and $\qN_{Y} A \subseteq X$ the sets defined as:

\begin{align*}
\qM_Y A &\eqdef \{x\in X \mid \text{$A_x$ is comeager}\} \\
\qN_Y A &\eqdef \{x\in X \mid \text{$A_x$ has Lebesgue measure $1$}\} 
\end{align*}

If the space $Y$ is known from the context, we skip the subscript $Y$. We refer to the operations $\qM$ and $\qN$ as the \emph{(Baire) Category quantifier} and the \emph{(Lebesgue) Measure quantifier}, respectively.

\begin{rem}
\label{remark_measurability}
It is consistent with $\textnormal{ZFC}$ that, e.g., there exists a projective set $A$ (even a $\asigma{2}$ set) which does not have the Baire property and is not Lebesgue measurable. However, the sets $\qN_Y A$ and $\qM_Y$ are well-defined even if the set $A$ is not Lebesgue measurable or it does not have the Baire property. 
\end{rem}

The following classical result states that the above operations preserve the class of Borel sets.

\begin{thm}
If $A\subseteq X\times Y$ is a Borel set then $\qM_Y A$ and $\qN_Y A$ are Borel sets.
\end{thm}
A proof of the above theorem can be obtained from~\cite[Theorem~29.22]{Kechris} for $\qM$ and~\cite[Theorem~29.26]{Kechris} for $\qN$, which are more generally stated for analytic sets, and an application of Souslin theorem~\cite[Theorem~14.11]{Kechris}.


\subsection{Alternating automata}

Given a finite set $X$, we denote with $\mathcal{DL}(X)$ the set of expressions $e$ generated by the grammar
\[e::= x \mid e \wedge e \mid e \vee e\quad \text{for $x\in X$}\]
An \emph{alternating tree automaton} over a finite alphabet $\Sigma$ is a tuple $\mathcal{A}=\langle \Sigma,Q, q_\init, \delta, \mathcal{F})$ where $Q$ is a finite \emph{set of states}, $q_\init\in Q$ is the \emph{initial state}, $\fun{\delta}{Q\times \Sigma}{\mathcal{DL}( \{\dL,\dR\}\times Q)}$ is the \emph{alternating transition function}, $\mathcal{F}\subseteq\powerset(Q)$ is a set of subsets of $Q$ called the \emph{Muller condition}.
The  Muller condition $\mathcal{F}$ is called a \emph{parity condition} if there exists a \emph{parity assignment} $\fun{\pi}{Q}{\w}$ such that:
$\mathcal{F}= \big\{F\subseteq Q \mid (\max_{q\in F}\pi(q))\textnormal{ is even}\}$. 

An alternating automaton $\mathcal{A}$ over an alphabet $\Sigma$ defines, or ``accepts'', a set of $\Sigma$-trees. The \emph{acceptance} of a tree $t\in \trees{\Sigma}$ is defined via a two-player ($\eve$ and $\adam$) game of infinite duration denoted as $\mathcal{A}(t)$. The configurations of the game $\mathcal{A}(t)$ are of the form $\langle u,q\rangle$ or $\langle u,e\rangle$ with $u\in\{\dL,\dR\}^{\ast}$, $q\in Q$, and $e\in \mathcal{DL}( \{\dL,\dR\}\times Q)$ (we can implicitly restrict to the set of sub-formulae of the formulae appearing in the alternating transition function $\delta$).

The game $\mathcal{A}(t)$ starts in the configuration $\langle\epsilon, q_\init\rangle$. Game configurations of the form $\langle u,q\rangle$, including the initial state, have only one successor configuration, to which the game progresses automatically. The successor state is $\langle u,e\rangle$ with $e=\delta(q,a)$, where  $a= t(u)$ is the labelling of the vertex $u$ given by $t$.  The dynamics of the game 
 at configurations $\langle u,e\rangle$ depends on the possible shapes of $e$.
If $e = e_1\vee e_2$, then Player~$\eve$ moves either to $\langle u,e_1\rangle$ or $\langle u,e_2\rangle$. If $e=e_1\wedge e_2$, then Player~$\adam$ moves either to $\langle u,e_1\rangle$ or $\langle u,e_2\rangle$. If $e=(\dL,q)$ then the game progresses automatically to the state $\langle u\dL, q\rangle$. Lastly, if $e=(\dR,q)$ the game progresses automatically to the state $\langle u\dR, q\rangle$. 
Thus a play in the game $\mathcal{A}(t)$ is a sequence $\Pi$ of configurations, that can look like
\[
\Pi = (\langle \epsilon,q_\init\rangle, \dots, \langle \dL, q_1\rangle, \dots,\langle \dL\dR, q_2\rangle,\dots,  \langle \dL\dR\dL, q_3\rangle, \dots, \langle \dL\dR\dL\dL, q_4\rangle,\dots),\]
where the dots represent part of the play in configurations of the form $\langle u,e\rangle$. Let $\Inf(\Pi)$ be the set of automata states $q\in Q$ occurring infinitely often in the configurations of the form $\langle u,q\rangle$ of $\Pi$. We then say that the play $\Pi$ of $\mathcal{A}(t)$ is winning for $\eve$, if $\Inf(\Pi)\in\mathcal{F}$. The play $\Pi$ is winning for $\adam$ otherwise. The \emph{language} of $\Sigma$-trees defined by $\Aa$ denoted $\lang(\Aa)$ is the collection $\{t\in\trees{\Sigma} \mid \text{$\eve$ has a winning strategy in the game $\mathcal{A}(t)$}\}$. 

An alternating automaton for $\w$-words is defined analogously, except that $\fun{\delta}{Q\times \Sigma}{\mathcal{DL}(Q)}$. In that case the the first coordinate of configurations of the game $\Aa(\alpha)$ is always a natural number, the game moves from a configuration $(n,q)$ to the configuration $(n+1,e)$ with $e=\delta(q,a)$, where $a=\alpha(n)$ is the labelling of the position $n$ given by $\alpha$. In that case $\lang(\Aa)\subseteq\Sigma^\w$.

\paragraph*{\bf Types of automata.} An alternating automaton $\Aa$ is \emph{non-deterministic} if all the transitions of $\Aa$ are of the form
\[
\big((\dL,q_{1,\dL}) {\land} (\dR,q_{1,\dR})\big)\ \lor\ 
\big((\dL,q_{2,\dL}) {\land} (\dR,q_{2,\dR})\big)\ \lor\ 
\ldots\ \lor\ 
\big((\dL,q_{k,\dL}) {\land} (\dR,q_{k,\dR})\big).
\]
An alternating automaton $\Aa$ is a \emph{game automaton} if all the transitions of $\Aa$ are of the form $(\dL,q_{\dL}) \xOper (\dR,q_{\dR})$ with ${\xOper}={\lor}$ or ${\land}$. 
An alternating automaton $\Aa$ is \emph{deterministic} if it is a game automaton with ${\xOper}={\land}$ in all transitions.


\section{Syntax and Semantics of Monadic Second Order Logic}
\label{section_mso_1}

In this section we define the syntax and the semantics of the  Monadic Second-order logic interpreted over the linear order of natural numbers (denoted $\SoneS$) and over the full binary tree (denoted $\StwoS$). This material is entirely standard. A  detailed exposition can be found in~\cite{thomas96}.

\subsection{Logic \texorpdfstring{$\SoneS$}{S1S} and \texorpdfstring{$\w$}{omega}-words.} 

The syntax of $\SoneS$ formulas is given by the  grammar

$$
\phi\, ::=\,   x<y \mid x\in X \mid \neg \phi \mid \phi_1 \vee \phi_2 \mid \forall x.\, \phi \mid \forall X.\, \phi
$$

where $x,y$ and $X,Y$ range over  countable sets of \emph{first-order variables} and \emph{second-order variables}, respectively. 

The semantics of  $\SoneS$ formulas on the structure $(\w,{<})$ is defined as expected by interpreting the formula $\forall X.\, \phi$ as quantifying over all subsets $X\subseteq\omega$ and the relation (${\in}$) as the membership relation between elements $x\in \w$ and sets $X\subseteq \omega$.

For example the formula  
\[
\exists X.\, \forall x.\, \big(
(x\in X) \Rightarrow (\exists y.\,  y\in X \wedge x<y)\big)
\]
expresses the existence of an unbounded set $X$ of natural numbers.

We write $\phi(\vec{x},\vec{X})$ to specify that the first-order and second-order variables of $\phi$ are $\vec{x}$ and $\vec{X}$, respectively.  

We can always assume that the free variables of a formula $\phi$ are all second-order variables, i.e., $\phi=\phi(\emptyset,\vec{X})$. Indeed if $\phi(\{x_1,\dots, x_n\},\vec{Y})$ is not of this form, we can consider the semantically equivalent formula $\psi(\emptyset, \{X_1,\ldots, X_n\} \cup\vec{Y})$, having second-order $\vec{X}$ variables (constrained to be singletons) simulating first-order ones.

Each set $X\subseteq \omega$ can be identified as the corresponding $\w$-word over the alphabet $\Sigma=\{0,1\}$. Hence a $\SoneS$ formula $\phi(X_1,\ldots,X_n)$ defines a collection of $n$-tuples of $\w$-words over $\Sigma=\{0,1\}$ or, equivalently, a collection of $\w$-words over the alphabet $\Gamma=\{0,1\}^n$.
The \emph{language defined by $\phi$}, denoted by $\lang(\phi)\subseteq \Gamma^\w$ is the set of $\omega$-words over $\Gamma$ satisfying $\phi$.

The following theorem and proposition are well known.
\begin{thm}
The theory of $\SoneS$ is decidable. The expressive power of $\SoneS$ is effectively equivalent with alternating / non-deterministic / deterministic Muller automata.
\end{thm}

\begin{prop}
For every $\SoneS$ formula $\phi$, the language $\lang(\phi)$ is a Borel subset of $\Gamma^\omega$.
\end{prop}

Hence for every $\SoneS$ formula $\phi$, the language $\lang(\phi)$ has the Baire property and is Lebesgue measurable.

\subsection{Logic \texorpdfstring{$\StwoS$}{S2S} and infinite trees} We now introduce, following a similar presentation, the syntax and the semantics of the \MSO\ logic on the full binary tree (\StwoS).

The \emph{full binary tree} is the structure $(\{\dL,\dR\}^\ast,\succL,\succR)$ where $\succL$ and $\succR$ are binary relations on $\{\dL,\dR\}^\ast$ relating a vertex with its left and right child, respectively:

\[
\succL = \{ (w,w\dL) \mid w\in \{\dL,\dR\}^\ast \} \qquad \succR = \{ (w,w\dR) \mid w\in \{\dL,\dR\}^\ast \}
\]

The syntax of $\StwoS$ formulas is given by  the grammar

\[
\phi\, ::=\,   \succL(x,y) \mid \succR(x,y)  \mid x\in X \mid \neg \phi \mid \phi_1 \vee \phi_2 \mid \forall x.\, \phi \mid \forall X.\, \phi
\]

where $x,y$ and $X,Y$ range over countable sets of \emph{first-order variables} and \emph{second-order variables}, respectively. 

The semantics of $\StwoS$ formulae on the full binary tree is defined by interpreting the formula $\forall X.\, \phi$ as quantifying over all subsets $X\subseteq\{\dL,\dR\}^\ast$ and  $x\in X$ as the membership relations between vertices $x\in \{\dL,\dR\}^\ast$ and sets $X\subseteq \{\dL,\dR\}^\ast$.

For example the formula  
\[
\exists X.\, \forall x.\, \big(
(x\in X) \Rightarrow (\exists y.\,  y\in X \wedge   \succL(x,y))\big)
\]
expresses the existence of a set of vertices which is closed under taking left children.

As for \SoneS, we can assume without loss of generality that the free variables of a $\StwoS$ formula $\phi$ are all second-order. 

Each set $X\subseteq \{\dL,\dR\}^\ast$ can be identified as a tree $t\in \trees{\Sigma}$ over the alphabet $\Sigma=\{0,1\}$. Hence a $\StwoS$ formula $\phi(X_1,\ldots,X_n)$ defines a collection of $n$-tuples of trees in $\trees{\{0,1\}}$ or, equivalently, a collection of trees in $\trees{\Gamma}$ over the alphabet $\Gamma=\{0,1\}^n$.  The set of $t\in \trees{\Gamma}$ satisfying the formula $\phi$ is called the \emph{language defined by $\phi$} and is denoted by $\lang(\phi)$.

The following theorem and proposition are well known.
\begin{thm}
The theory of $\StwoS$ is decidable. The expressive power of \MSO\ over infinite trees is effectively equivalent with alternating / non-deterministic  Muller automata. The expressive power of deterministic automata is strictly weaker.
\end{thm}

\begin{prop}
For every $\StwoS$ formula $\phi$, the language $\lang(\phi)$ is a $\adelta{2}$ subset of $\trees{\Gamma}$.
\end{prop}

It has recently been shown that $\lang(\phi)$ is always Lebesgue measurable and has the Baire property~\cite{GMMS2014}.

\begin{prop}
For every $\StwoS$ formula $\phi$ the language $\lang(\phi)$ has the Baire property and is Lebesgue measurable.
\end{prop}


\section{S1S extended with the category quantifier}
\label{sec:s1s-category}
In this section we define the syntax and semantics of  the logic $\SoneS$ extended with the \emph{category quantifier} ($\qM$). The syntax of $\SoneSM$ extends that of $\SoneS$ with the new second-order quantifier $\qM$ as follows:

\[
\phi \ \ ::= \ \  x<y \mid x\in X \mid \neg \phi \mid \phi_1 \vee \phi_2 \mid \forall x.\, \phi \mid \forall X.\, \phi \mid \qM X.\, \phi
\]

The semantics of the category quantifier is  specified as follows.

\[
\lang\big(\qM X.\, \phi(X,Y_1,\dots, Y_n)\big)  = \Big\{
w\in\Gamma^\w \mid \text{$\{ w' \in \{ 0,1\}^\w \mid \phi(w',w)\}$ is comeager}\Big\}
\]
where $\Gamma=\{0,1\}^n$. 

In other words, an $n$-tuple $w=(w_1,\ldots,w_n)$, with $w_i\in \{0,1\}^\w$,  satisfies the formula $\qM X.\, \phi(X,Y_1,\ldots, Y_n) $ if  for a large (comeager) collection of $w'\in \{0,1\}^\w$, the extended tuple $(w',w)$ satisfies $\phi(X,Y_1,\dots, Y_n) $. Informally,  $w$ satisfies $\qM X.\, \phi(X,\vv{Y})$ if ``for almost all''  $w' \in \{0,1\}^\w$, the tuple $(w',w)$ satisfies $\phi$. The set defined by $\qM X.\, \phi(X,\vv{Y})$ can be illustrated as in Figure~\ref{figure_large_section}, as the collection of tuples $w$ having a \emph{large} (comeager) section $\phi(X,w)$.

The main result about $\SoneSM$ is the following quantifier elimination theorem which appears to be folklore. We provide a complete proof of this theorem for two reasons: first, instead of just simulating the Banach--Mazur game over $\w$-words, we want to inductively remove all the meager quantifiers from the formula; second, the proof presented here is a good base for understanding the more difficult proof from Section~\ref{sec:s2s-category}.

\begin{thm}
\label{thm:sonesMelim}
For every $\SoneSM$ formula $\phi$ one can effectively construct a semantically equivalent $\SoneS$ formula $\psi$.
\end{thm}

\begin{cor}
The theory of $\SoneSM$ is decidable.
\end{cor}

The rest of this section is devoted to an inductive proof of Theorem~\ref{thm:sonesMelim}. For that sake it is enough to show how to eliminate one application of the quantifier $\qM$. Consider a formula of $\SoneSM$ of the form $\qM X.\, \psi$, where $\psi$ is a formula of $\SoneS$. We will prove that there exists a formula $\psi'$ of $\SoneS$ that is equivalent to $\qM X.\, \psi$. We start by translating $\psi$ into an equivalent deterministic Muller automaton $\Aa$. Then we apply the following technical proposition that provides us an alternating Muller automaton $\Bb$, recognising $\w$-words that satisfy $\qM X.\, \psi$. Finally, we can turn $\Bb$ into a non-deterministic Muller automaton and express $\lang(\Bb)$ by an $\SoneS$ formula $\psi'$.

\begin{prop}
\label{pro:s1s-m-cons}
Let $\Aa$ be a deterministic Muller automaton over $\w$-words, with the alphabet $\Sigma\times \Gamma$. Consider the language $\qM \lang(\Aa)$ that contains an $\w$-word $\alpha\in\Sigma^\w$ if and only if the set
\begin{equation}
\label{eq:def-of-lang-B-words}
\big\{\beta\in\Gamma^\w\mid (\alpha,\beta)\in\lang(\Aa)\}
\end{equation}
is comeager.

Then, $\qM \lang(\Aa)$ is $\w$-regular and one can effectively construct an alternating Muller automaton $\Bb$ for this language. Additionally, the size of the automaton $\Bb$ is polynomial in the size of $\Aa$.
\end{prop}

Fix a deterministic automaton $\Aa$ over a product alphabet $\Sigma\times\Gamma$. Let $Q$ be the set of states of $\Aa$ and $\fun{\delta^{\Aa}}{Q\times \Sigma\times\Gamma}{Q}$ be the transition function: for a triple $(q,a,b)\in Q\times \Sigma\times\Gamma$ it assigns the successive state $\delta^{\Aa}(q,a,b)\in Q$ that should be reached after reading a letter $(a,b)$ from the state $q$.

We will construct an alternating automaton $\Bb$ over the alphabet $\Sigma$. Intuitively, $\Bb$ will simulate the Banach--Mazur game over $\qM \lang(\Aa)$ in which a new $\w$-word $\beta\in\Gamma^\w$ is constructed. During that play, the automaton $\Aa$ will be run over the pair of $\w$-words $(\alpha,\beta)\in(\Sigma\times\Gamma)^\w$ to verify if this pair belongs to $\lang(\Aa)$. To achieve this, the set of states of $\Bb$ will keep track of both, the current state of $\Aa$ and the player in charge of construction of $\beta$. More formally, let the set of states of $\Bb$ be $Q\times\{\eve,\adam\}$.



\usetikzlibrary{decorations.pathreplacing}
\usetikzlibrary{automata,positioning}
\usetikzlibrary{decorations.pathmorphing}
\usetikzlibrary{positioning}

\tikzstyle{ubrace} = [draw, thick, decoration={brace, mirror, raise=0.0cm}, decorate,
    every node/.style={anchor=north, yshift=-0.1cm}]
\tikzstyle{rbrace} = [draw, thick, decoration={brace, mirror, raise=0.0cm}, decorate,
    every node/.style={anchor=west, xshift= 0.1cm}]

\tikzstyle{obrace} = [draw, thick, decoration={brace, raise=0.0cm}, decorate,
    every node/.style={anchor=south, yshift= 0.1cm}]
\tikzstyle{lbrace} = [draw, thick, decoration={brace, raise=0.0cm}, decorate,
    every node/.style={anchor=east, xshift=-0.1cm}]

\usetikzlibrary{shapes.geometric,arrows,calc,decorations.pathmorphing}

\tikzstyle{ePls} = [draw, diamond, minimum width=20pt, minimum height=25pt]
\tikzstyle{aPls} = [draw, rectangle,minimum width=18pt, minimum height=18pt]
\tikzstyle{fake} = [inner sep=60pt]
\tikzstyle{trans} = [draw, >=latex, ->, every node/.style={auto, scale=0.8}, pos=0.8]
\tikzstyle{state} = [scale=0.8]

\tikzstyle{dot} = [draw,shape=circle,fill, minimum size=1mm, inner sep=0pt,outer sep=0pt]

\newcommand{\Qx}{0.8}
\newcommand{\Qy}{1.8}

\tikzstyle{flow}  = [inner sep=0cm, node distance=0cm and 0cm]

\newcommand{\bigDots}[1]{
	\node[align=center, anchor=center, flow, scale=2.5, yshift=3pt] at #1 {$\vdots$};
}

\newcommand{\tikzEvalFloat}[2]{\pgfmathparse{#2}{\global\edef#1{\pgfmathresult}}}

\newcommand{\iniTranF}[1]{
	\tikzEvalFloat{\x}{#1}
	\node[fake] (q\x1a1) at ($(q\x) + (0,-\x+2) + (3, +0.5)$) {};
	\node[fake] (q\x1a2) at ($(q\x) + (0,-\x+2) + (3, -0.5)$) {};

	\draw (q\x) edge[draw]  (q\x1a1);
	\draw (q\x) edge[draw]  (q\x1a2);
	\bigDots{(1.5,+\Qy+\Quy - 1.2*\x+0.6-\x*0.3+0.6)}
}

\newcommand{\iniTranR}[2]{
	\tikzEvalFloat{\x}{#1}
	\node[#2Pls] (q\x1a1) at ($(q\x) + (-0.4,0) + (3, +0.5)$) {};
	\node[#2Pls] (q\x1a2) at ($(q\x) + (-0.4,0) + (3, -0.5)$) {};

	\draw (q\x) edge[trans] node {$a$} (q\x1a1);
	\draw (q\x) edge[trans] node {$a'$} (q\x1a2);
}

\newcommand{\whole}[3]{
\tikzEvalFloat{\Quy}{#1}

\draw (0-\Qx,-\Qy+\Quy) rectangle (0+\Qx, \Qy+\Quy);
\draw (9-\Qx,-\Qy+\Quy) rectangle (9+\Qx, \Qy+\Quy);

\foreach \x in {1,2,3} {
	\node[dot] (q\x) at (0.4,+\Qy+\Quy-\x*1.2 + 0.6) {};
	\node[state] at ($(q\x) + (-0.55, 0.0)$) {$(q_\x, #3)$};

	\node[dot] (p\x#2) at (8.6, +\Qy+\Quy-\x*1.2 + 0.6) {};
	\node[state] at ($(p\x#2) + (+0.55, 0.0)$) {$(q_\x, #3)$};
}
}

\newcommand{\vvY}{1.2}

\begin{figure}
\centering
\begin{tikzpicture}
\whole{+3.2}{e}{\eve}

\foreach \x in {1,3} {
	\iniTranF{\x};
}

\foreach \x in {2} {
	\iniTranR{\x}{e}
}

\node[fake] (ea1b1) at ($(q21a1) + (3, +1.5)$) {};
\node[fake] (ea1b2) at ($(q21a1) + (3, +0.5)$) {};

\node[ePls] (ea2b1) at ($(q21a2) + (3, +\vvY)$) {};
\node[ePls] (ea2b2) at ($(q21a2) + (3, -\vvY)$) {};

\draw (q21a2) edge[trans] node {$b$} (ea2b1);
\draw (q21a2) edge[trans] node {$b'$} (ea2b2);

\draw (q21a1) edge[draw] (ea1b1);
\draw (q21a1) edge[draw] (ea1b2);

\bigDots{(4.1,+\Qy+\Quy-1.0)}

\whole{-3.2}{a}{\adam}

\foreach \x in {2,3} {
	\iniTranF{\x};
}

\foreach \x in {1} {
	\iniTranR{\x}{a}
}

\node[aPls] (aa1b1) at ($(q11a1) + (3, +\vvY)$) {};
\node[aPls] (aa1b2) at ($(q11a1) + (3, -\vvY)$) {};

\node[fake] (aa2b1) at ($(q11a2) + (3, -0.5)$) {};
\node[fake] (aa2b2) at ($(q11a2) + (3, -1.5)$) {};

\draw (q11a1) edge[trans] node {$b$} (aa1b1);
\draw (q11a1) edge[trans] node {$b'$} (aa1b2);

\draw (q11a2) edge[draw] (aa2b1);
\draw (q11a2) edge[draw] (aa2b2);

\bigDots{(4.1,-\Qy+\Quy+1.0+1.2)}

\node[fake] (ea1b1) at ($(aa1b2) + (3, +0.5)$) {};
\node[fake] (ea1b2) at ($(aa1b2) + (3, -0.5)$) {};

\draw (aa1b2) edge[draw] (ea1b1);
\draw (aa1b2) edge[draw] (ea1b2);

\bigDots{(7.1,-\Qy+-3.0+2.0)}

\node[fake] (ea1b1) at ($(ea2b1) + (3, +0.5)$) {};
\node[fake] (ea1b2) at ($(ea2b1) + (3, -0.5)$) {};

\draw (ea2b1) edge[draw] (ea1b1);
\draw (ea2b1) edge[draw] (ea1b2);

\bigDots{(7.1,+\Qy+3.0-0.9)}

\draw (aa1b1) edge[trans, pos=0.6] node{$\eve$} (p3e);
\draw (aa1b1) edge[trans, pos=0.6] node{$\adam$} (p3a);

\draw (ea2b2) edge[trans, pos=0.6] node{$\eve$} (p2e);
\draw (ea2b2) edge[trans, pos=0.6] node{$\adam$} (p2a);

\draw (9.0, +3.0+\Qy) edge[trans, =>, double, bend right] (0.0, +3.0+\Qy);
\draw (9.0, -3.0-\Qy) edge[trans, =>, double, bend left] (0.0, -3.0-\Qy);

\node[scale=0.8] at (11.5, +3.2) {$\delta^{\Aa}(q_2, a', b')=q_2$};
\node[scale=0.8] at (11.5, +2.0) {$\delta^{\Aa}(q_1, a, b)=q_3$};

\node[scale=0.8] at (11.5, -3.2) {$\delta^{\Aa}(q_2, a', b')=q_2$};
\node[scale=0.8] at (11.5, -4.4) {$\delta^{\Aa}(q_1, a, b)=q_3$};

\draw (-1.1, +3.2-\Qy) edge[lbrace] node[scale=0.8]{$Q\times\{\eve\}$} ++(0.0,2.0*\Qy);

\draw (-1.1, -3.2-\Qy) edge[lbrace] node[scale=0.8]{$Q\times\{\adam\}$} ++(0.0,2.0*\Qy);

\draw (-0.75, -7.5) edge[ubrace] node[scale=0.7]{transition starts in $(q,r)$} ++(1.5,0.0);
\draw (0.25, -6.5) edge[ubrace] node[scale=0.7]{letter $a\in\Sigma$ is given} ++(2.5,0.0);
\draw (3.25, -7.0) edge[ubrace] node[scale=0.7]{Player $r$ chooses $b\in\Gamma$} ++(2.5,0.0);
\draw (6.25, -6.5) edge[ubrace] node[scale=0.7]{Player $r$ chooses $r'\in\{\eve,\adam\}$} ++(2.5,0.0);
\draw (8.25, -7.5) edge[ubrace] node[scale=0.7]{transition ends in $\big(\delta^{\Aa}(q,a,b), r'\big)$} ++(1.5,0.0);
\end{tikzpicture}
\caption{The structure of the automaton $\Bb$. The transitions are depicted with the convention that choices of \eve (i.e., ${\lor}$) are represented by diamonds; and the choices of \adam (i.e., ${\land}$) are represented by squares. The states of the automaton are drawn twice, on the left and right edge of the picture for the sake of readability, the double arrows show that we actually go back to the left copy. The phases of each transition (i.e.~rounds of the acceptance game) are explained by the braces below the picture.}
\label{fig:auto-for-words}
\end{figure}
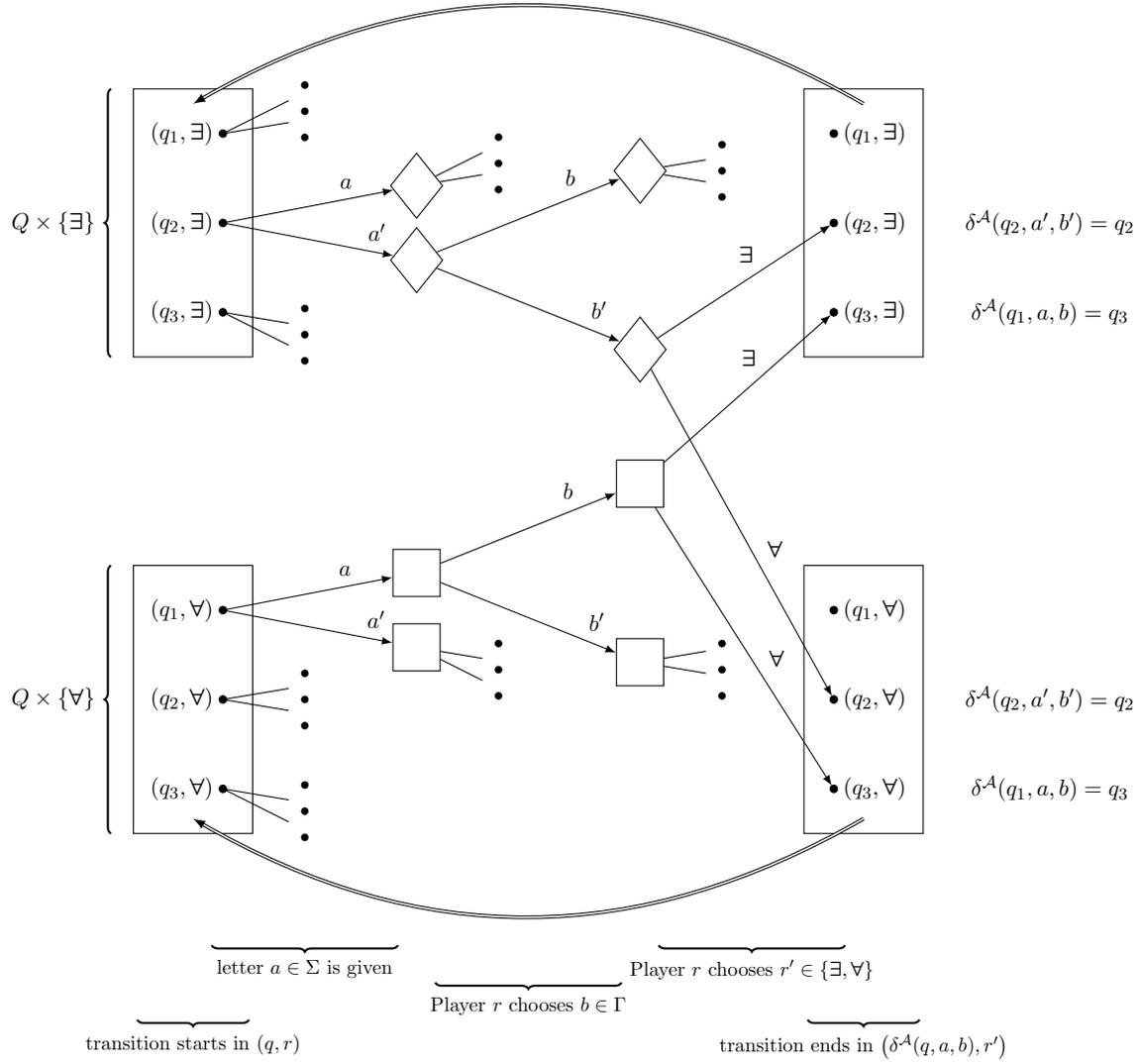

The initial state of $\Bb$ is $(q_\init, \adam)$. Consider a state $(q,r)$ of $\Bb$ and a letter $a\in\Sigma$. The transition $\delta^{\Bb}(q,r,a)$ is supposed to allow the player $r$ to choose a letter $b\in\Gamma$ and the next player $r'\in\{\eve,\adam\}$. After that choices are done, the successive state of the automaton should be $(q,r')$ where $q=\delta^{\Aa}(q,a,b)$. The structure of the automaton $\Bb$ is depicted in Figure~\ref{fig:auto-for-words}. More formally, we achieve that by putting:
\begin{equation}
\delta^{\Bb}(q,r,a)\eqdef\XOper_{b\in\Gamma}\, \XOper_{r'\in\{\eve,\adam\}}\, (q,r'),
\end{equation}
where ${\xOper}$ is ${\vee}$ if $r=\eve$ and ${\wedge}$ otherwise, and $\delta^{\Aa}(q,a,b)=q$.

The acceptance condition of the automaton $\Bb$ will be the following Muller condition: a set of states $F\subseteq Q\times\{\eve,\adam\}$ belongs to $\Ff^{\Bb}$ if either:
\begin{enumerate}
\item all the states $(q,r)\in F$ satisfy $r=\adam$ (i.e.~from some point on $\adam$ has always chosen $r'=r=\adam$,
\item or $F$ contains both states of the form $(q,\eve)$ and $(q',\adam)$; and additionally the set of states $\{q\in Q\mid \text{$(q,r)\in F$ for some $r\in\{\eve,\adam\}$}\}$ is accepting for $\Aa$ (i.e.~belongs to $\Ff^{\Aa}$).
\end{enumerate}
In other words, if any of the players from some point on chooses $r'=r$ then he or she looses. Otherwise, the game is won by \eve if the sequence of states visited is accepting in $\Aa$.

The following lemma proves correctness of our construction and thus concludes the proof of Theorem~\ref{thm:game-meager}.

\begin{lem}
\label{lem:s1s-m-corr}
The automaton $\Bb$ accepts an $\w$-word $\alpha\in\Sigma^\w$ if and only if the language from~\eqref{eq:def-of-lang-B-words} is comeager.
\end{lem}

Informally, the core of the above lemma is the observation that for a fixed $\alpha\in\Sigma^\w$ there is a bijection between:
\begin{itemize}
\item the plays of the Banach--Mazur game over the language~\eqref{eq:def-of-lang-B-words} (denoted $\BM$);
\item the plays of the acceptance game of $\Bb$ over $\alpha$, in which infinitely many times $r'\neq r$.
\end{itemize}

To make this relationship more formal, we will show how to simulate plays of $\BM$ using the acceptance game $\Bb(\alpha)$. Consider an $\w$-word $\alpha\in\Sigma^\w$ and assume that $P\in\{\eve,\adam\}$ wins the acceptance game of $\Bb$ over $\alpha$ (i.e.~$\alpha\in\lang(\Bb)$ iff $P=\eve$). We will prove that $P$ wins $\BM$.

Fix a strategy $\sigma_P$ of $P$ in the acceptance game $\Bb(\alpha)$. We will inductively describe how $P$ can simulate this strategy in the game $\BM$. Assume that we are at the beginning of the $n$th round of $\BM$ and the current configuration of this game is $(s_n,r_n)$.

\begin{clm}
\label{cl:s1s-mea-inv}
During the simulated play of $\BM$ the following invariant will be preserved: there exists a finite play of $\Bb(\alpha)$, that is consistent with $\sigma_P$ and in this play, after reading the first $|s_n|$ symbols of $\alpha$ the reached configuration has the form $\langle |s_n|, (q_n,r_n)\rangle$ for some $q_n\in Q$ (and $r_n$ as above).
\end{clm}

We will now show how to play one round of $\BM$ while still preserving the invariant from Claim~\ref{cl:s1s-mea-inv}. Consider the following two possibilities for $r_n\in\{\eve,\adam\}$.

\paragraph*{\bf The case $r_n\neq P$.} In that case it is the opponent of $P$ (denoted $\bar{P}$) that chooses the successive word $s_{n+1}$ in $\BM$. Assume that $\bar{P}$ has played $s_{n+1}=s_n w$ for some non-empty word $w=b_0b_1\ldots b_k$. Consider the $k{+}1$ rounds of $\Bb(\alpha)$ starting from $\langle |s_n|, (q_n,r_n)\rangle$ in which $\bar{P}$ plays the successive letters $b_0,b_1,\ldots,b_k$. In the rounds $0,1,\ldots,k-1$ he chooses $r'=r=\bar{P}$ and in the last round he chooses $r'\neq r$ such that $r' = P$. After that $k{+}1$ rounds the simulated play of $\Bb(\alpha)$ has reached a configuration of the form $\langle |s_{n+1}|, (q_{n+1},r_{n+1})\rangle$ with $r_{n+1}=P$, and the invariant is satisfied.

\paragraph*{\bf The case $r_n= P$.} Now consider the more involved case where this is Player~$P$ who chooses a word $s_{n+1}$. Consider the successive rounds of the game $\Bb(\alpha)$ after the round that ended in $\langle |s_n|, (q_n,r_n)\rangle$. Since $r_n=P$ we know that it will be Player~$P$ who will propose successive letters $b_0,b_1,\ldots$ and players $r'$ in each round. Let $b_0,b_1,\ldots$ be the finite or infinite sequence of letters played by $P$ according to $\sigma_P$ until $P$ chooses to set $r'\neq r = P$ for the first time. Since $\sigma_P$ is winning, we know that $P$ cannot keep $r'=r=P$ forever, otherwise $P$ would lose the infinite play. Thus, the sequence $b_0,b_1,\ldots,b_k$ is finite and after playing $b_k$ Player~$P$ has chosen $r'\neq r= P$. Let $s_{n+1}=s_n b_0 b_1\cdots b_k$ be the response of $P$ in the game $\BM$. After that, the configuration of $\BM$ is $(s_{n+1},\bar{P})$ and the configuration of $\Bb(\alpha)$ is of the form $\langle |s_{n+1}|, (q_{n+1}, r_{n+1})\rangle$ with $r_{n+1}=\bar{P}$, so the invariant is satisfied.

This finishes the proof of Claim~\ref{cl:s1s-mea-inv}. What remains is to prove that the described simulation strategy is winning for $P$ in $\BM$.

\paragraph*{\bf Why $P$ wins.} We need to prove that the produced word $\beta=\bigcup_{n\in\w} s_n$ satisfies $(\alpha,\beta)\in\lang(\Aa)$ if and only if $P=\eve$. But we know that the strategy $\sigma_P$ is winning for $P$, therefore the simulated play of $\Bb(\alpha)$ must be winning for $P$. Since we have guaranteed that infinitely many times $r'\neq r$, it means that the sequence of visited states of $\Aa$ is accepting iff $P=\eve$. Thus, $(\alpha,\beta)\in\lang(\Aa)$ iff $P=\eve$, what means that $P$ has won the considered play of $\BM$. 

This concludes the proof of Lemma~\ref{lem:s1s-m-corr} as well as Theorem~\ref{thm:sonesMelim}.


\section{S1S extended with the measure quantifier}
\label{sec:s1s-measure}
In this section we define the syntax and semantics of  the logic $\SoneS$ extended with the \emph{measure quantifier} ($\qN$). The syntax of $\SoneSN$ extends that of $\SoneS$ with the new second-order quantifier $\qN$ as follows:
\[
\phi \ \ ::= \ \  x<y \mid x\in X \mid \neg \phi \mid \phi_1 \vee \phi_2 \mid \forall x.\, \phi \mid \forall X.\, \phi \mid \forall^{=1} X.\, \phi
\]
The semantics of the measure is  specified as follows.
\[
\lang\Big(
\qN X.\, \phi(X,Y_1,\dots, Y_n)\Big)  = \Big\{
w\in\Gamma^\w \mid
	\mu_w\big( \{ w^\prime \in \{ 0,1\}^\w \mid  \phi(w^\prime,w) \textnormal{ holds} \} \big) = 1      \Big\}
\]
where $\Gamma=\{0,1\}^n$ and $\mu_\w$ is the Lebesgue measure on $\{0,1\}^\omega$.

In other words, a $n$-tuple  $w=(w_1,\dots,w_n)$, with $w_i\in \{0,1\}^\w$,  satisfies the formula $\forall^\ast X.\, \phi(X,Y_1,\dots, Y_n) $ if  for a large (having measure $1$) collection of $w^\prime\in \{0,1\}^\w$, the extended tuple $(w^\prime,w)$ satisfies $\phi(X,Y_1,\dots, Y_n) $.

Once again, informally,  $w$ satisfies $\qN X.\, \phi(X,\vv{Y})$ if ``for almost all''  $w^\prime \in \{0,1\}^\omega$, the tuple $(w^\prime,w)$ satisfies $\phi$. In the logic $\SoneSN$ the  ``almost all'' means for all but a negligible (having measure $0$) set.

The main result of this section is that, unlike the case of $\SoneSM$, the theory of $\SoneSN$ is undecidable (Theorem~\ref{undec:s1s:measure} below). This result was first proved in~\cite{MM2016LFCS} by means of a reduction to the emptiness problem of \emph{probabilistic B{\"u}chi automata} which is known to be undecidable from~\cite{BGB2010}.

Here we propose a different proof based on the recent result of~\cite{bpt2016} stating the logic $\SoneSU$, introduced by Boja{\'n}czyk in~\cite{bojanczyk2004}, has an undecidable theory. We recall that the syntax of $\SoneSU$ extends that of $\SoneS$ with the new second-order \emph{unbounding quantifier} ($\qU X.\,\phi$) whose semantics is specified as follows:

\[
\lang\Big(
\qU X.\, \phi(X,Y_1,\dots, Y_n)\Big)  = \Big\{
w\in\Gamma^\w \mid
\forall n. \exists w^\prime\in \{0,1\}^\w. \big( n<|w^\prime| <\infty \wedge 
	 \phi(w^\prime,w) \textnormal{ holds}\big)\Big\}
\]
where $\Gamma=\{0,1\}^n$ and $n<|w^\prime|<\infty$ means that $w^\prime$, seen as a subset of $\mathbb{N}$, is finite and has at least $n$ elements.

\begin{thm}
\label{thm:msou}
For every formula $\varphi$ of $\SoneSU$ there effectively exists a formula $\varphi'$ of $\SoneSN$ such that $\lang(\varphi)=\lang(\varphi')$.
\end{thm}
The above reduction, together
~with~\cite{bpt2016,hummel_msou}, give us the following corollary.
\begin{cor}
\label{undec:s1s:measure}
The logic $\SoneSN$ has an undecidable theory. Furthermore, for each positive number $k\in\w$ there exists a $\SoneSN$ formula $\phi(X)$ such that $\lang(\phi)\subseteq \{0,1\}^\w$ does not belong to the $\adelta{k}$ class of the projective hierarchy.
\end{cor}


The rest of this section is devoted to proving Theorem~\ref{thm:msou}. We will assume that using $\SoneS$ logic we can freely quantify over $\w$-words over fixed alphabets --- this can be simulated by quantifying over tuples of sets representing positions that bear a given letter.

The first step of our proof is standard: we prove that instead of using the $\qU$ quantifier in $\SoneSU$, one can use a specific predicate.

\begin{defi}
Let $\alpha\in\{0,1,R\}^\w$ be an infinite word. We say that $\alpha$ is \emph{unbounded} (denoted $U(\alpha)$) if $\alpha$ contains infinitely many letters $R$ and, for every $n$ there exists a pair of consecutive letters $R$ such that in-between them there are at least $n$ letters $1$.
\end{defi}

Before moving forward, we will introduce a couple of useful concepts. 
\begin{defi}
Fix $\alpha\in\{0,1,R\}^\w$ and assume that $\alpha$ contains infinitely many letters $R$. 
\begin{enumerate}
\item A set $B$ of consecutive positions $\{i,i+1,\ldots,j\}$ of $\alpha$ is called a \emph{block} if: $\alpha$ contains no $R$ at positions $\{i,\ldots,j-1\}$; $i=0$ or $\alpha(i-1)=R$; and $\alpha(j)=R$. In other words, a block is a maximal set of successive positions of $\alpha$ such that the only $R$ that appears among them is the last position in the block.
\item Since blocks are disjoint, we can identify an arbitrary set of blocks with its union $S$ --- a subset of $\w$.
\item If $B$ is a block, by the \emph{value of $B$} we denote the number $v$ of positions of $\alpha$ in $B$ that are labeled by $1$. $v$ is always a non-negative integer.
\end{enumerate}
\end{defi}

\begin{lem}
The logic $\SoneSU$ is effectively expressively equivalent to the logic $\SoneS$ equipped with the $U$ predicate defined above.
\end{lem}

The above lemma is a standard technical adjustment of the $\qU$ quantifier, see for instance the notion of \emph{sequential witness} from Lemma~5.5 of~\cite{bojanczyk_bounds}. For the sake of completeness, we sketch a proof of it.

\begin{proof}
Clearly the $U$ predicate can easily be expressed by the $\qU$ quantifier. Now consider an application of a $\qU$ quantifier of the form $\qU X.\, \phi(X)$. We claim that the following formula is equivalent to $\qU X.\, \phi(X)$
\begin{align*}
\psi\eqdef \exists \alpha.\, U(\alpha)\wedge \forall B.\,& \big(\text{$B$ is a block of $\alpha$}\big)\Rightarrow \\
& \Big[\exists X.\, \text{$X$ is finite, $\phi(X)$, and } \forall x\in B.\, \alpha(x)=1\Rightarrow x\in X\Big].
\end{align*}
Clearly, if $\psi$ holds then  $\qU X.\, \phi(X)$ holds as well (the positions labeled $1$ in the blocks of $\alpha$ witness that). The other direction is also easy, it is enough to construct a witness $\alpha$. We proceed inductively, at $n$th step taking a finite set $X$ of cardinality large enough to guarantee that it has at least $n$ elements greater in the order ${\leq}$ from all the previously considered positions. These $n$ new positions will be labeled $1$ in the $n$th block of $\alpha$.
\end{proof}

In the presence of the above lemma, Theorem~\ref{thm:msou} follows directly from the following lemma.

\begin{lem}
\label{lem:coding_in_u}
There exists a formula $\psi_U$ of $\SoneSN$ over the alphabet $\{0,1,R\}$ such that $\alpha\in\{0,1,R\}^\w$ satisfies $\psi$ if and only if $U(\alpha)$ holds.
\end{lem}

\begin{proof}
Consider the following formula $\psi_U$:
\begin{align*}
\psi_U(\alpha)&\eqdef\exists S.\,\text{$S$ is an infinite set of blocks in $\alpha$}\ \land\\
& \quad\quad\lnot\big[\forall^{=1} X.\, \exists B.\, \text{$B\subseteq S$ and $B$ is a block in $\alpha$}\ \land\\
& \quad\quad\quad\quad \forall x.\, \text{if $x\in B$ and $\alpha(x)=1$ then $x\in X\big]$.}
\end{align*}

\begin{clm}
\label{cl:msou-main}
The formula $\psi_U(\alpha)$ holds if and only if $U(\alpha)$ holds.
\end{clm}

Without loss of generality we can restrict to $\w$-words $\alpha$ containing infinitely many letters $R$ (otherwise both properties are false).

To simplify our notation, let us denote by $\varphi(X,B)$ the property that: for every $x\in B$ such that $\alpha(x)=1$ we have $x\in X$. Consider a fixed block $B$ of $\alpha$ that has value $v$ and let $X\subseteq\w$ be a randomly chosen set. Then, the probability that $\varphi(X,B)$ holds is equal to $2^{-v}$. Now, if $S$ is an infinite set of blocks of $\alpha$ with values $v_0,v_1,\ldots$, the probability that for a random set $X\subseteq \w$ some block of $S$ satisfies $\varphi(X,B)$ is equal to
\begin{equation}
\label{eq:value-of-S}
1-\prod_{n\in\w}\ \big(1-2^{-v_n}\big)
\end{equation}
For an illustration of an $\w$-word $\alpha$ and its division into blocks, see Figure~\ref{fig:alpha-and-blocks}.

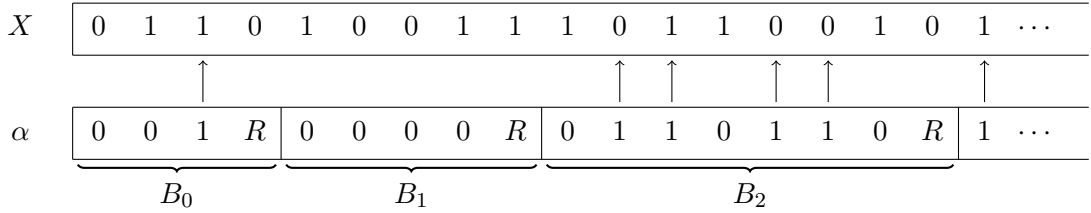
\begin{figure}
\centering
\newcommand{\XyCoor}{3}
\newcommand{\AyCoor}{1}
\newcommand{\maxX}{19.5}

\begin{tikzpicture}[scale=\textwidth/22cm]

\draw[black] (-0.5, \XyCoor + 0.5) -- ++(\maxX, 0);
\draw[black] (-0.5, \XyCoor + 0.5) -- ++(0, -1);
\draw[black] (-0.5, \XyCoor - 0.5) -- ++(\maxX, 0);

\draw[black] (-0.5, \AyCoor + 0.5) -- ++(\maxX, 0);
\draw[black] (-0.5, \AyCoor + 0.5) -- ++(0, -1);
\draw[black] (-0.5, \AyCoor - 0.5) -- ++(\maxX, 0);

\node[toC] at (-1.5, \XyCoor) {$X$};
\node[toC] at (-1.5, \AyCoor) {$\alpha$};

\tikzEvalInt{\x}{0}

\foreach \letter in {0, 0, 1, R, 0, 0, 0, 0, R, 0, 1, 1, 0, 1, 1, 0, R,1} {
	\node[toC] at (\x, \AyCoor) {$\letter$};
	\tikzEvalInt{\x}{\x+1}
}

\tikzEvalInt{\x}{0}

\foreach \letter in {0, 1, 1, 0, 1, 0, 0, 1, 1, 1, 0, 1, 1, 0, 0, 1, 0,1} {
	\node[toC] at (\x, \XyCoor) {$\letter$};
	\tikzEvalInt{\x}{\x+1}
}

\foreach \x in {3, 8, 16} {
	\draw[black] (\x + 0.5, \AyCoor + 0.5) -- ++(0, -1);
}

\foreach \x in {2, 10, 11, 13, 14, 17} {
	\draw[black, ->] (\x, \AyCoor + 0.6) -- (\x, \XyCoor - 0.6);
}

\node[toC] at (18, \XyCoor) {$\cdots$};
\node[toC] at (18, \AyCoor) {$\cdots$};

\draw (-0.4, \AyCoor-0.6) edge[ubrace] node{$B_0$} ++(3.8,0);
\draw ( 3.6, \AyCoor-0.6) edge[ubrace] node{$B_1$} ++(4.8,0);
\draw ( 8.6, \AyCoor-0.6) edge[ubrace] node{$B_2$} ++(7.8,0);

\end{tikzpicture}
\caption{A word $\alpha$ divided into blocks $B_0$, $B_1$, \ldots The values of the blocks are $v_0=1$, $v_1=0$, $v_2=4$, \ldots The formula $\varphi(X,B)$ holds for $B_0$ and $B_1$ but not for $B_2$.}
\label{fig:alpha-and-blocks}
\end{figure}

We are now in position to prove Claim~\ref{cl:msou-main}. First consider $\alpha$ such that it is not the case that $U(\alpha)$ holds. Therefore, there is a global bound $M$ such that if $B$ is a block in $\alpha$ then the value of $B$ is at most $M$. Let $S$ be any infinite set of blocks in $\alpha$ as in the formula $\psi_U$. By~\eqref{eq:value-of-S}, the probability that some block $B$ of $S$ satisfies $\varphi(X,B)$ equals $1-\prod_{n\in\w} \big(1-2^{-v_n}\big)\geq 1-\prod_{n\in\w} \big(1-2^{-M}\big)=1-0=1$. Therefore, $\psi_U(\alpha)$ is false.

Now consider the opposite case that $U(\alpha)$ holds. We will prove that $\psi_U(\alpha)$ holds. By the assumption, for every $M$ there exists a block $B$ in $\alpha$ of value greater than $M$. Let $0<c_0<c_1<\ldots<1$ be a sequence of numbers such that $\prod_{n\in\w} c_n > 0$. For each $n$, there exists a block $B_n$ of $\alpha$ of value $v_n$ such that $1-2^{-{v_n}}>c_n$. Let $S$ be the union of the blocks $B_n$, clearly $S$ is infinite. Take a random set $X\subseteq\w$. By~\eqref{eq:value-of-S}, the probability that some block $B$ of $S$ satisfies $\varphi(X,B)$ equals $1-\prod_{n\in\w} \big(1-2^{-v_n}\big)\leq 1-\prod_{n\in\w} c_n<1$. Therefore, $\psi_U(\alpha)$ holds.

This concludes the proof of Claim~\ref{cl:msou-main} and also the proof of Theorem~\ref{thm:msou}.
\end{proof}

\section{S2S extended with the category quantifier}
\label{sec:s2s-category}
In this section we consider the extension of the monadic second order logic of the full binary tree (\StwoS) with the category quantifier.

The syntax of $\StwoSM$ extends that of $\StwoS$ with the new second-order quantifier $\qM$ as follows:

\[
\phi\, ::=\,   \succL(x,y) \mid \succR(x,y)  \mid x\in X \mid \neg \phi \mid \phi_1 \vee \phi_2 \mid \forall x.\, \phi \mid \forall X.\, \phi \mid \qM X.\, \phi
\]

The semantics of the category quantifier is specified as follows.

\[
\lang\Big(
\qM X.\, \phi(X,Y_1,\dots, Y_n)\Big)  = \Big\{
t\in\trees{\Gamma} \mid
	\textnormal{ the set } \{ t^\prime\in\trees{\{0,1\}}\mid  \phi(t^\prime,t) \textnormal{ holds} \} \big) \textnormal{ is comeager}     \Big\}
\]
where $\Gamma=\{0,1\}^n$.

Our main theorem is the following result, showing that the category quantifier preserves \StwoS-definability when the inner formula can be recognised by a game automaton.

\begin{thm}
\label{thm:game-meager}
Let $\Aa$ be a game automaton over an alphabet $\Sigma\times \Gamma$. Consider the language $\qM \lang(\Aa)$ that contains a tree $t\in\trees{\Sigma}$ if and only if the set
\begin{equation}
\label{eq:def-of-lang-B-trees}
\big\{t'\in\trees{\Gamma}\mid (t,t')\in\lang(\Aa)\}
\end{equation}
is comeager.

Then, $\qM \lang(\Aa)$ is regular and one can effectively construct an alternating Muller tree automaton $\Bb$ for this language. Additionally, the size of the automaton $\Bb$ is polynomial in the size of $\Aa$.
\end{thm}

It is quite clear that if $L$ is a regular language of trees, then the Banach--Mazur game $\BM(L)$ (see Definition~\ref{def:bm-trees}) cannot be directly encoded in the \StwoS logic because a strategy of a player in this game is a complex object (at least it assigns prefixes to prefixes).

Fix a game automaton $\Aa$ over a product alphabet $\Sigma\times\Gamma$. Let $Q$ be the set of states of $\Aa$ and $\fun{\delta^{\Aa}}{Q\times \Sigma\times\Gamma}{\mathcal{DL}( \{\dL,\dR\}\times Q)}$ be the transition function: for a triple $(q,a,b)\in Q\times \Sigma\times\Gamma$ it assigns an expression of the form $(\dL,q_\dL)\xOper(\dR,q_\dR)$ with ${\xOper}$ either $\lor$ (a transition of $\eve$) or $\land$ (a transition of $\adam$).

We will construct an alternating automaton $\Bb$ over the alphabet $\Sigma$ that will read a tree $t\in\trees{\Sigma}$ and verify that the language~\eqref{eq:def-of-lang-B-trees} is comeager. The construction follows the same lines as the proof of Proposition~\ref{pro:s1s-m-cons}. The additional difficulty here is that we deal with a model of automata that is a bit stronger than deterministic ones. Therefore, we need to take some additional care to resolve the game arising from the semantics of these automata.

Intuitively, $\Bb$ will simulate synchronously two games: the Banach--Mazur game over $\qM \lang(\Aa)$ in which a new tree $t'\in\trees{\Gamma}$ is constructed; and the acceptance game of $\Aa$ over the constructed pair of trees $(t,t')$. This parallel execution of the two games will be visible in the transitions of $\Bb$. Let the set of states of $\Bb$ keep track of both, the current state of $\Aa$ and the player in charge of construction of the tree, i.e., the set of states of $\Bb$ is $Q\times\{\eve,\adam\}$.

The initial state of $\Bb$ is $(q_\init, \adam)$. Consider a state $(q,r)$ of $\Bb$ and a letter $a\in\Sigma$. The transition $\delta^{\Bb}(q,r,a)$ is supposed to give the following choices to the players:
\begin{itemize}
\item Player~$r$ chooses a letter $b\in\Gamma$,
\item Player~$r$ chooses the next player, $r'\in\{\eve,\adam\}$,
\item the owner of the transition $\delta^{\Aa}(q,a,b)=(\dL,q_\dL)\yOper(\dR,q_\dR)$ in $\Aa$ resolves this formula by choosing a direction $d\in\{\dL,\dR\}$.
\end{itemize}
After that choices are done, the transition that is taken should be $\langle d,(q_d,r')\rangle$, i.e.~we should move in the direction $d$ to the new state $(q_d,r')$. More formally, we achieve that by putting:
\begin{equation}
\delta^{\Bb}\langle(q,r),a\rangle\eqdef\XOper_{b\in\Gamma}\, \XOper_{r'\in\{\eve,\adam\}}\, \YOper_{d\in \{\dL,\dR\}}\, \langle d,(q_d,r')\rangle,
\end{equation}
where ${\xOper}$ is ${\vee}$ if $r=\eve$ and ${\wedge}$ otherwise, and $\delta^{\Aa}\langle q,(a,b)\rangle=(\dL,q_\dL)\yOper(\dR,q_\dR)$ (i.e.~the last boolean operator in $\delta^{\Bb}\langle (q,r),a\rangle$ is the same as the one in $\delta^{\Aa}\langle q,(a,b)\rangle$).

The acceptance condition of the automaton $\Bb$ will be the same as in Section~\ref{sec:s1s-category}: a set of states $F\subseteq Q\times\{\eve,\adam\}$ belongs to $\Ff^{\Bb}$ if either:
\begin{enumerate}
\item all the states $(q,r)\in F$ satisfy $r=\adam$ (i.e.~from some point on $\adam$ has always chosen $r'=r=\adam$,
\item or $F$ contains both states of the form $(q,\eve)$ and $(q',\adam)$; and additionally the set of states $\{q\in Q\mid \text{$(q,r)\in F$ for some $r\in\{\eve,\adam\}$}\}$ is accepting for $\Aa$ (i.e.~belongs to $\Ff^{\Aa}$).
\end{enumerate}
In other words, if any of the players from some point on chooses $r'=r$ then he or she looses. Otherwise, the game is won by \eve if the sequence of states visited is accepting in $\Aa$.

The following lemma proves correctness of our construction and thus concludes the proof of Theorem~\ref{thm:game-meager}. Notice that, except the choice of directions $d$, the construction is almost identical to the one presented in Section~\ref{sec:s1s-category}.

\begin{lem}
The automaton $\Bb$ accepts a tree $t\in\trees{\Sigma}$ if and only if the language from~\eqref{eq:def-of-lang-B-trees} is comeager.
\end{lem}

Before giving a formal (but technical) proof of this lemma, we will provide an overview. Assume that a player $P$ has a winning strategy $\sigma_P$ in the acceptance game $\Bb(t)$ of $\Bb$ over a tree $t\in\trees{\Sigma}$. By the definition we know that $t\in\lang(\Bb)$ iff $P=\eve$. Our aim is to simulate the strategy $\sigma_P$ as a strategy of $P$ in the Banach--Mazur game over the language~\eqref{eq:def-of-lang-B-trees} (we denote this game $\BM$ for the rest of the proof). This will imply that: $t\in\lang(\Bb)$ iff $P=\eve$ iff $\eve$ wins $\BM$ iff the language~\eqref{eq:def-of-lang-B-trees} is comeager.

Since the domain of all our trees is the same --- $\{\dL,\dR\}^\ast$, we can imagine that during a play of $\BM$, the players write the respective prefixes $\parfun{s}{\{\dL,\dR\}^\ast}{\Gamma}$ on top of the given tree $t\in\trees{\Sigma}$. More formally, for $\parfun{s}{\{\dL,\dR\}^\ast}{\Gamma}$ let us define a prefix $\parfun{t\otimes s}{\{\dL,\dR\}^\ast}{\Sigma\times\Gamma}$ as follows: let $\dom(t\otimes s)\eqdef\dom(s)$ and for $u\in\dom(s)$ let $\big(s\otimes t\big)(u)\eqdef\big(t(u),s(u)\big)$.

We will use the strategy $\sigma_P$ to help the player~$P$ win $\BM$ by the following two requirements:
\begin{enumerate}
\item Player~$P$ will provide the letters of $s_n$ and the moments to finish the current prefix according to his choices in the transitions of $\Bb$,
\item Player~$P$ will separately store some information about his choices regarding directions in a data structure called $\tau$.
\end{enumerate}

\usetikzlibrary{decorations.pathreplacing}
\usetikzlibrary{automata,positioning}
\usetikzlibrary{decorations.pathmorphing}
\usetikzlibrary{decorations.markings}

\tikzstyle{tau}=[draw=black, ultra thick]
\tikzstyle{endpoint}=[draw, circle, inner sep=1.8pt]

\tikzstyle{rbrace} = [draw, thick, decoration={brace, mirror, raise=0.0cm}, decorate, every node/.style={anchor=west, xshift= 0.1cm, scale=1.0}]

\tikzstyle{lbrace} = [draw, thick, decoration={brace, raise=0.0cm}, decorate,
    every node/.style={anchor=east, xshift=-0.1cm, scale=1.0}]

\newcommand{\tD}{0.3}
\newcommand{\tV}{0.6}

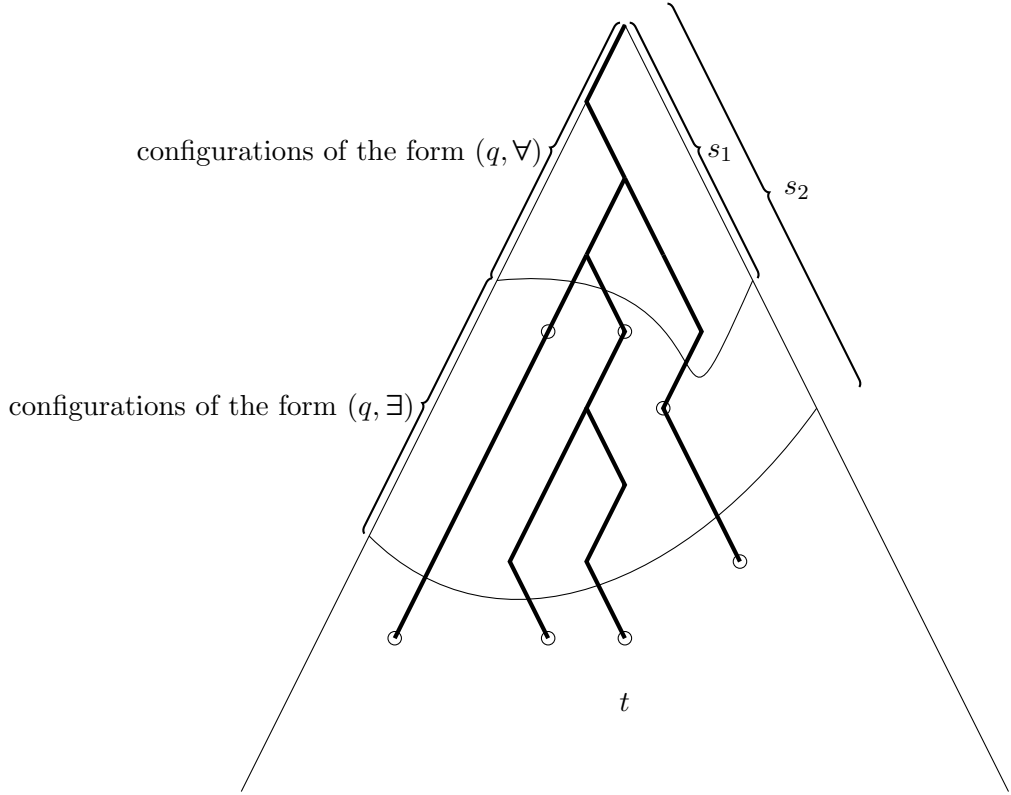
\begin{figure}
\centering
\begin{tikzpicture}[scale=1.7]

\draw[draw=black] (-3,-6) -- (0,0) -- (+3, -6);

\draw (-1, -2) .. controls ++(+2.0,+0.2) and ++(-0.8,-1.8) .. (1,-2);

\draw (-2, -4) .. controls ++(+1,-1.0) and ++(-1,-1.4) .. (1.5,-3);

\draw[tau] (0,0) -- ++(-\tD,-\tV) -- ++(+\tD,-\tV) ;

\draw[tau] (0,-2*\tV) -- ++(-\tD, -\tV);
\draw[tau] (0,-2*\tV) -- ++(+\tD, -\tV);

\draw[tau] (-\tD,-3*\tV) -- ++(-\tD, -\tV) -- ++(-\tD, -\tV) -- ++(-\tD, -\tV) -- ++(-\tD, -\tV) -- ++(-\tD, -\tV);
\draw[tau] (-\tD,-3*\tV) -- ++(+\tD, -\tV) -- ++(-\tD, -\tV) -- ++(+\tD, -\tV) -- ++(-\tD, -\tV) -- ++(+\tD, -\tV);
\draw[tau] (-\tD, -5*\tV) -- ++(-\tD, -\tV) -- ++(-\tD, -\tV) -- ++(+\tD, -\tV);
\draw[tau] (\tD,-3*\tV) -- ++(+\tD, -\tV) -- ++(-\tD, -\tV) -- ++(\tD, -\tV) -- ++(\tD, -\tV) ;

\node[endpoint] at (-2*\tD, -4*\tV) {};
\node[endpoint] at (+0*\tD, -4*\tV) {};
\node[endpoint] at (+1*\tD, -5*\tV) {};
\node[endpoint] at (+3*\tD, -7*\tV) {};
\node[endpoint] at (-6*\tD, -8*\tV) {};
\node[endpoint] at (-2*\tD, -8*\tV) {};
\node[endpoint] at (+0*\tD, -8*\tV) {};

\draw (1.0,-2.0) edge[rbrace, decoration={raise=2pt}] node{$s_1$} (0.0,-0);
\draw (1.5,-3.0) edge[rbrace, decoration={raise=18pt}] node[xshift=17pt, yshift=10pt]{$s_2$} (0.0,-0);

\draw (-1.0,-2.0) edge[lbrace, decoration={raise=2pt}] node {configurations of the form $(q,\adam)$} (0.0,-0);

\draw (-2,-4) edge[lbrace, decoration={raise=2pt}] node {configurations of the form $(q,\eve)$} (-1.0,-2.0);

\node[scale=1.0] at (0, -5.3) {$t$};

\end{tikzpicture}
\caption{An illustration of our simulation procedure. We are right after the second round of $\BM$. The currently played prefixes are $s_0=\emptyset$, $s_1$ and $s_2$. The boldfaced subtree of $t$ is our data structure $\tau_2$ (it contains $\tau_1$). The nodes in circles are the end-points of $\tau_1$ and $\tau_2$. While we are still in $s_1$, all the configurations visited in the simulated plays of $\Aa(t,t')$ are of the form $\langle u,(q,\adam)\rangle$. Then, within $s_2$ we ensure that the configurations have the form $\langle u, (q,\eve)\rangle$ etc. The branching in $\tau_1$ and $\tau_2$ occurs exactly in those places where the respective transition of the game automaton $\Aa$ is controlled by the opponent of our Player~$P$. When the respective transition belongs to $P$, the choice of direction is resolved immediately. In that case we do not care about the content of the played prefixes in the skipped subtree (outside the respective $\tau_n$).}
\label{fig:game-invariant}
\end{figure}

Then, we will prove that a play of $\BM$ is won by Player~$P$ by noticing that the finally obtained data structure $\tau$ is a winning strategy of $P$ in the acceptance game $\Aa(t,t')$ of the game automaton $\Aa$ over the product tree $(t,t')$ with $t'=\bigcup_{n\in\w} s_n$.

The following claim describes the invariants of our construction, see Figure~\ref{fig:game-invariant} for an illustration.

\begin{clm}
\label{cl:inv-game-meager}
Assume that the current configuration of $\BM$ is $(s_n,r_n)$. In that case, there must exist: a set $\tau_n\subseteq\{\dL,\dR\}^\ast$ of nodes and an assignment $\pi_n$ that for every node $u\in\tau_n$ gives a finite play $\pi_n(u)$ of the acceptance game of $\Aa$ over $t\otimes s_n$. The set $\tau_n$ and the assignment $\pi_n$ need to extend the previous ones $\tau_{n-1}$ and $\pi_{n-1}$ and additionally the following invariants are required:
\begin{enumerate}
\item $\epsilon\in\tau_n$; $\tau_n$ is prefix-closed; and $\pi_n(\epsilon)$ consists of the initial configuration $\langle \epsilon, (q_\init, \adam)\rangle$ of $\Aa(t\otimes s_n)$.
\item If $u\in\tau_n$ and $u\notin s_n$ then no sequence extending $u$ belongs to $\tau_n$ (we call such $u$ an \emph{end-point} of $\tau_n$).
\item If $ud\in\tau_n$ for some $u$, $d$, then the play $\pi_n(ud)$ extends the play $\pi_n(u)$ by a round played according to $\sigma_P$ in which the chosen letter $b$ is $s_n(u)$; the chosen player $r'$ equals the previous player $r$ iff $ud$ is not an end-point for one of $\tau_0$, $\tau_1$, \ldots, $\tau_n$; and the chosen direction is $d$.
\label{it:inv-rounds}
\item Take $u\in\tau_n$. Then the play $\pi_n(u)$ needs to end in a configuration $\langle u, (q,r)\rangle$.
\label{it:inv-confs}
\item Take $u\in\tau_n$ that is not an end-point of $\tau_n$ and assume that the play $\pi_n(u)$ ends in a configuration of the form $\langle u, (q,r)\rangle$. If $\delta^{\Aa}\big(q, t(u), s_n(u)\big)$ is a transition of Player~$\bar{P}$ then both $u\dL$ and $u\dR$ belong to $\tau_n$; otherwise exactly one of $u\dL$, $u\dR$ belongs to $\tau_n$ (the choice of $d$ will depend on the strategy $\sigma_P$ of $P$).
\end{enumerate}
\end{clm}

Notice that Invariants~\eqref{it:inv-rounds} and~\eqref{it:inv-confs} guarantee that the last configuration of the play $\pi_n(u)$ for $u$ an end-point of $\tau_n$ is of the form $\langle u, (q, r_n)\rangle$ --- the position in the tree is $u$ and the players switch exactly when moving to an end-point of one of $\tau_0$, \ldots, $\tau_n$.

Initially, for $n=0$, we have $s_0=\emptyset$ and all the invariants are satisfied by $\tau_0=\{\epsilon\}$ and $\pi_0$ that maps $\epsilon$ to the initial configuration of the acceptance game of $\Aa$ (i.e.~$\langle\epsilon,(q_\init,\adam)\rangle$).

We will now describe how to inductively preserve the invariants from Claim~\ref{cl:inv-game-meager}. Consider an $n$th round of the game $\BM$. Its initial configuration is $(s_n,r_n)$ with a set $\tau_n$ and assignment $\pi_n$. There are two cases depending whether $r_n=P$ or not.

\paragraph*{\bf Simulation: the case of $r_n=\bar{P}$.} First assume that $r_n\neq P$, i.e.~the considered round of $\BM$ is controlled by our opponent. Assume that in this round the player $r_n$ plays a prefix $s_{n+1}$ that extends $s_n$. Thus, the round ends in the configuration $(s_{n+1}, P)$. We construct $\tau_{n+1}$ and $\pi_{n+1}$ as follows: start from $\tau_{n+1}:=\tau_n$, $\pi_{n+1}:=\pi_n$ and inductively extend $\tau_{n+1}$ and $\pi_{n+1}$ for every end-point $u$ of $\tau_{n+1}$ that belongs to $s_{n+1}$. Consider a round of the acceptance game of $\Aa$ over $t\otimes s_{n+1}$ right after the last configuration of $\pi_{n+1}(u)$. We know that this configuration is of the form $\langle u,(q,\bar{P})\rangle$.

Assume that in this round $\bar{P}$ chooses as the letter $b=s_{n+1}(u)$ and as the successive player $r'=\bar{P}$ if $u$ is not a leaf of $s_{n+1}$ and $r'=P$ otherwise. Now consider the following cases for the rest of this round:
\begin{itemize}
\item If the transition $\delta^{\Aa}\big(q,t(u),s_{n+1}(u)\big)$ belongs to $\bar{P}$ then for both $d\in\{\dL,\dR\}$ we add $ud\in \tau_{n+1}$ and define $\pi_{n+1}(ud)$, assuming that $\bar{P}$ played $d=\dL$ and $d=\dR$ respectively.
\item Otherwise, the strategy $\sigma_P$ chooses some direction $d\in\{\dL,\dR\}$. We then add $ud\in \tau_{n+1}$ and define $\pi_{n+1}(ud)$, assuming that $P$ played $d$. In that case $u\bar{d}\notin\tau_{n+1}$.
\end{itemize}
Clearly by the definition all the invariants are satisfied in this case.

\paragraph*{\bf Simulation: the case of $r_n=P$.} Now take the more involved case when $r_n=P$, i.e.~the considered round of $\BM$ is controlled by Player $P$ which strategy $\sigma_P$ we simulate. We will define $s_{n+1}$, $\tau_{n+1}$, and $\pi_{n+1}$ inductively, starting from the end-points of $\tau_n$. For nodes outside the constructed set $\tau_{n+1}$ the letters of $s_{n+1}$ are arbitrary, i.e.~we can assume that for every $ud$ such that $u\in\tau_{n+1}$ or $u\in s_n$ but $ud\notin \tau_{n+1}$ we let $s_{n+1}(ud)=b_0$ for some fixed letter $b_0\in\Gamma$. In that case $ud$ will be a leaf of $s_{n+1}$.

We start from $\tau_{n+1}:=\tau_n$, $\pi_{n+1}:=\pi_n$. Let $u$ be an end-point of $\tau_{n+1}$ such that the last configuration of the play $\pi_{n+1}(u)$ is of the form $\langle u, (q, r)\rangle$. If $r\neq P$ then we finish this branch of construction, letting $u\notin s_{n+1}$ be an end-point of $\tau_{n+1}$. If $r=P$ it means that $u$ will not be an end-point of $\tau_{n+1}$. Consider the successive round after the play $\pi_{n+1}(u)$ in which $P$ plays according to $\sigma_P$. First, $P$ chooses a letter $b\in \Gamma$ and a player $r'\in\{\eve,\adam\}$. We can immediately define $s_{n+1}(u)=b$. Now consider the following cases for the rest of this round:
\begin{itemize}
\item If the transition $\delta^{\Aa}\big(q,t(u),b\big)$ belongs to $\bar{P}$ then for both $d\in\{\dL,\dR\}$ we add $ud\in \tau_{n+1}$ and define $\pi_{n+1}(ud)$, assuming that $\bar{P}$ played $d=\dL$ and $d=\dR$ respectively.
\item Otherwise, the strategy $\sigma_P$ chooses some direction $d\in\{\dL,\dR\}$. We then add $ud\in \tau_{n+1}$ and define $\pi_{n+1}(ud)$, assuming that $P$ played $d$. In that case $u\bar{d}\notin\tau_{n+1}$.
\end{itemize}

If the above inductive procedure ends after finitely many steps, we obtain a finite prefix $s_{n+1}$ together with $\tau_{n+1}$ and $\pi_{n+1}$. Notice that for every end-point $u$ of $\tau_{n+1}$ the last configuration of the play $\pi_{n+1}(u)$ is of the form $\langle u,(q, \bar{P})\rangle$. Thus, the invariant is satisfied.

Consider the opposite case that the procedure runs indefinitely. In that case, by K\"onig's Lemma, there is an infinite path in the constructed set $\tau_{n+1}$. This path corresponds to an infinite play of the acceptance game of $\Bb$ over $t$ in which Player~$P$ keeps $r'=r=P$ constantly equal $P$ from some point on. This contradicts the assumption that $P$ plays according to a winning strategy $\sigma_P$, as such a play is losing for Player~$P$.

This way we have managed to play a consecutive round of $\BM$ while preserving the invariants. Therefore, by induction on $n$ we can construct a strategy of Player~$P$ in $\BM$.

\paragraph*{\bf Why $P$ wins?} What remains to prove is that any play of $\BM$ in which $P$ plays according to this simulation strategy is winning. Take such a play and consider $t'=\bigcup_{n\in\w} s_n$, $E=\bigcup_{n\in\w} \tau_n$, $\Pi =\bigcup_{n\in\w} \pi_n$. It remains to prove that $(t,t')\in \lang(\Aa)$ if and only if $P=\eve$. We achieve that by proving that $E$ encodes a winning strategy of $P$ in the acceptance game of $\Aa$ over $(t, t')$. Clearly, by the structure of all the sets $\tau_n$, their union $E$ encodes the following strategy of $P$: stay in the nodes of $(t,t')$ that belong to $E$. Take any infinite branch $\beta$ contained in $E$. Notice that there is a unique play $\pi$ of the acceptance game of $\Bb$ over $t$ that is the limit of the plays $\Pi(u)$ for $u\prec \beta$. Thus, since we considered plays according to a winning strategy $\sigma_P$ of $P$, this play $\pi$ must be winning for $P$. Clearly, the sequence of players $r$ in this play does not have a limit. Thus, the sequence of visited states $q$ must be either accepting if $P=\eve$ or rejecting otherwise. This concludes the proof that the strategy represented by $E$ is winning for $P$. Therefore, $(t,t')\in\lang(\Aa)$ iff $P=\eve$ and thus $P$ wins $\BM$.


\section{S2S extended with the measure quantifier}
\label{sec:s2s-measure}
In this section we consider the extension of the monadic second order logic of the full binary tree (\StwoS) with the measure quantifier.

The definitions of the syntax and semantics of $\StwoS+\qN$ are similar to those given for $\SoneS+\qN$. The syntax of $\StwoSN$ extends that of $\StwoS$ with the new second-order quantifier $\qN$ as follows:

\[
\phi\, ::=\,   \succL(x,y) \mid \succR(x,y)  \mid x\in X \mid \neg \phi \mid \phi_1 \vee \phi_2 \mid \forall x.\, \phi \mid \forall X.\, \phi \mid \qN X.\, \phi
\]

The semantics of the measure quantifier is specified as follows.

\[
\lang\Big(
\qN X.\, \phi(X,Y_1,\dots, Y_n)\Big)  = \Big\{
t\in\trees{\Gamma} \mid
	\mu_t\big( \{ t^\prime\in\trees{\{0,1\}}\mid  \phi(t^\prime,t) \textnormal{ holds} \} \big) = 1      \Big\}
\]
where $\Gamma=\{0,1\}^n$ and $\mu_t$ is the Lebesgue measure on $\trees{\{0,1\}}$.

Once again, informally,  $t$ satisfies $\forall^{=1}X.\, \phi(X,\vv{Y})$ if ``for almost all''  $t^\prime$, the tuple $(t^\prime,t)$ satisfies $\phi$, where ``almost all'' means for all but a negligible (having measure $0$) set.

Our main result regarding the logic $\StwoS+\qN$ is the following.

\begin{thm}\label{thm:msoqn:undecidable}
The logic $\StwoS+\qN$ has an undecidable theory. 
\end{thm}

This is obtained as a direct corollary of Theorem~\ref{thm:msou} by interpreting the logic $\SoneS+\qN$ within $\SoneS+\qN$.

Recall that a standard interpretation of $\SoneS$ within $\StwoS$ is based on the identification of $(\mathbb{N},<)$ with the set of vertices $\{\dL\}^*$ (belonging to the leftmost branch of the full binary tree) ordered by the prefix relation, which are both easily definable in $\StwoS$. For example, the $\SoneS$ formula $\forall X.\phi(X)$ is translated to the $\StwoS$ formula 

$$
\forall X. \Big( \overline{\phi}( X\cap \{\dL\}^*)\big)
$$
where $\overline{\phi}(Y)$ is the translation of the simpler formula $\phi(Y)$.

Similarly, we define the translation of $\SoneS+\qN$ formulas of the form $\qN X.\, \phi$ as:

$$
\qN X. \Big( \overline{\phi}( X\cap \{\dL\}^*)\big)
$$
 
To check that this translation is correct it is sufficient to prove that the function $\pi$ 
$$
\pi(X)= X \cap \{\dL\}^* 
$$
mapping subsets of the full binary tree to subsets of the leftmost branch (which can be identified as $\omega$-words over the alphabet $\{0,1\}$) is Lebesgue measure preserving. That is, we need to show that $\pi$ is a continuous surjection, and this is obvious, and that for every Borel set $A\subseteq \{0,1\}^\w$ it holds that $\mu_t(\pi^{-1}(A)) = \mu_w(A)$. By regularity of the Lebesgue measure it is sufficient to  prove that this property holds for arbitrary basic clopen sets $A$.
These are sets of the form $A=U_{\vec{n}=\vec{b}}$, for tuples $\vec{n}=(n_1,\dots, n_k)\in \mathbb{N}^k$ and $(b_1,\dots, b_k)\in \{0,1\}^k$ where:
$$
U_{\vec{n}=\vec{b}} = \big\{ w\in\{0,1\}^{\omega} \mid \bigwedge^k_{i=1} w(n_i) = b_i\big\}
$$
By definition we have that $\mu_w(U_{\vec{n}=\vec{b}})=\frac{1}{2}^k$.

The preimage $\pi^{-1}(A)$ is the clopen set $B= \{ t\in\trees{\{0,1\}} \mid \bigwedge^k_{i=1} t(\dL^n_i) = b_i \}$ and, by definition of Lebesgue measure on trees, it holds that $\mu_t(B)=(\frac{1}{2})^k$, as desired. 


\section{S2S extended with the path-category quantifier}
\label{sec:s2s-category-pi}
As anticipated in the introduction, another interesting way to extend $\StwoS$ with variants of Friedman's Category and Measure quantifiers is to restrict the quantification to range over infinite branches (paths) of the full binary tree. 

In this section we consider the extension of $\StwoS$ with the category quantifier restricted to path, henceforth denoted by $\qMpi$.

We first recall the definition of paths in the full binary tree.

\begin{defi}
A set $X\subseteq \{\dL,\dR\}^*$ of vertices in the full binary tree is a \emph{path} if and only if:
\begin{enumerate}
\item $X$ contains the root $\epsilon$, and
\item if $v\in X$ and $w$ is a prefix of $v$ then $v\in X$,
\item if $v\in X$ then either $v\dL\in X$ or $v\dR\in X$, but not both. 
\end{enumerate}
We denote with $\mathcal{P}$ the collection of paths in the full binary tree. 
\end{defi}
Since every path is uniquely determined by an infinite sequence of directions ($\dL$ or $\dR$), there is a one-to-one correspondence between $\mathcal{P}$ and the  space $\{\dL,\dR\}^\omega$ which is homeomorphic to the Cantor space. 

It is easy to verify that $\mathcal{P}\subseteq \trees{\{0,1\}}$ is Lebesgue null ($\mu_t(\mathcal{P})=0$) and it is meager as a subset of $\trees{\{0,1\}}$. However, since $\mathcal{P}$ is homeomorphic to the Cantor space, it makes sense to consider the probability (i.e., Lebesgue measure $\mu_w$) and Baire category of subsets $A\subseteq\mathcal{P}\subseteq  \trees{\{0,1\}}$ \emph{relative to } $\mathcal{P}$.

This leads to the following definition of the Category-path quantifier $\qMpi$.

\[
\lang\Big(
\qMpi X.\, \phi(X,Y_1,\dots, Y_n)\Big)  = \Big\{
t\in\trees{\Gamma} \mid
	\textnormal{ the set } \{ p \in\mathcal{P}\mid  \phi(p,t^\prime) \textnormal{ holds} \} \big) \textnormal{ is comeager in } \mathcal{P}     \Big\}
\]
where $\Gamma=\{0,1\}^n$.

Our main result regarding the logic $\StwoS+\qMpi$ is the following quantifier elimination theorem.
\begin{thm}\label{qmpi_is_definable}
For every $\StwoS+\qMpi$ formula $\phi$ one can effectively construct an semantically equivalent $\StwoS$ formula $\psi$.
\end{thm}

\begin{proof}
A set of paths $B\subseteq \{\dL,\dR\}^\w$ is comeager if and only if there is a $\Gdelta$ set $G$ that is dense in $\{\dL,\dR\}^\w$ and $G\subseteq B$. A simple argument (see, e.g.,~\cite{Rabin69}) shows that $G\subseteq \{\dL,\dR\}^\w$ is a $\Gdelta$ set if and only if $G$ is of the form $[X]$ for a set $X\subseteq \{\dL,\dR\}^\ast$ where:
\[[X]\eqdef\{\pi\in\{\dL,\dR\}^\w\mid \text{for infinitely many vertices $v\in \pi$ we have $v\in X$}\}.\]

Therefore, $\qMpi \pi.\, \varphi(\pi)$ holds if and only if there exists a set of nodes $X$ of the full binary tree such that:
\begin{itemize}
\item for every node $v$ there exists a node $w$ such that $v\preceq w$ and $w\in X$ (i.e.~$[X]$ is dense in $\{\dL,\dR\}^\w$),
\item for every infinite branch $\pi$, if there are infinitely many $v\in X$ such that $v\prec \pi$ then $\varphi(\pi)$ holds (i.e., ~$[X]\subseteq \lang(\varphi)$)
\end{itemize}
The whole above property is easily \StwoS-definable.
\end{proof}


\section{S2S extended with the path-measure quantifier}
\label{sec:s2s-measure-pi}
Following the ideas presented in the previous section, we now study the extension of $\StwoS$ with the path-measure quantifier $\forall^{=1}_\pi$ whose semantics is defined as follows:

\[
\lang\Big(
\qNpi X.\, \phi(X,Y_1,\dots, Y_n)\Big)  = \Big\{
t\in\trees{\Gamma} \mid
	\mu_{w}\Big( \{ p \in\mathcal{P}\mid  \phi(p,t) \textnormal{ holds} \}  \big) = 1\Big\}
\]
where $\Gamma=\{0,1\}^n$. 

Therefore, the property of $\qNpi X.\, \phi$ holds if the property $\phi$ holds for almost all pahts $p$, where ``for almost all'' means having probability $1$ with respect to the Lebesgue measure $\mu_w$ on $\{\dL,\dR\}^\w$ (and thus not with respect to $\mu_t$ on $\trees{\{0,1\}}$).

It is not immediately clear from the previous definition if the quantifier $\qNpi$ can be expressed in $\StwoS+\qN$. Indeed, as already observed, the set $\mathcal{P}$ of paths is Lebesgue null with respect to the Lebesgue measure $\mu_t$ on $\trees{\{0,1\}}$, i.e.,  $\mu_t(\mathcal{P})=0$. Therefore  the naive definition 
$$\qNpi X.\, \phi( X)  = \qN X.\, \big( \text{``$X$ is a path''} \wedge \phi(X) \big)$$
does not work. Indeed the $\StwoS+\qN$ formula on the right always defines the empty set because the collection of $X\in\trees{\{0,1\}}$  satisfying the conjunction is a subset of $\mathcal{P}$ and therefore has $\mu_t$ measure $0$. 

Nevertheless the quantifier $\qNpi$ can be expressed in $\StwoS+\qN$ with a more elaborate encoding (Theorem~\ref{encodabilty_of_quantifier} below).  What is needed is a \StwoS  definable continuous and measure preserving function $f$ mapping trees $X\in\trees{\{0,1\}}$ to a paths $f(X)\in \mathcal{P}$. We now prove that such definable function exists.

\begin{defi}
\label{red_to_pi}
Define the binary relation $f(X, Y)$ on $\trees{\{0,1\}}$ by the following $\StwoS$ formula: 
\[\text{``$Y$ is a path''} \wedge  \forall  y\in Y.\, \exists z.\, ( \SuccL(y,z)  \wedge (z\in Y \Leftrightarrow y\in X)).\]
\end{defi}

\begin{lem}
\label{lemma_property_leadsto}
For every $X\in\trees{\{0,1\}}$ there exists exactly one $Y\in\mathcal{P}\subseteq\trees{\{0,1\}}$ such that $f(X,Y)$. Hence the relation $f$ is a function $\fun{f}{\trees{\{0,1\}}}{\mathcal{P}}$. Furthermore $f$ satisfies the following properties:
\begin{enumerate}
\item $f$ is a continuous and surjective function,
\item $f$ is measure preserving, i.e., for every Borel set $B\subseteq\mathcal{P}$ it holds that $\mu_t(f^{-1}(B))=    \mu_w(B)$.
\end{enumerate}
\end{lem}

\begin{proof}
The mapping $f$ is well defined in the sense that $Y$ is fully determined by the $\{0,1\}$-labeled tree $X$. Indeed,  from the condition 
\[\forall  y\in Y.\, \exists z.\, ( \SuccL(y,z) \wedge 
(z\in Y \Leftrightarrow y\in X)) \] in Definition~\ref{red_to_pi} it follows that for every vertex $y\in Y$, the set $X$ determines if the unique successor of $y$ in $Y$  is $y\dL$ or $y\dR$ (since $Y$ is a path, it contains either $y\dL$ or $y\dR$ and not both) depending on whether $y\!\in\!X$ or $y\!\not\in\! X$, respectively.  
Since a finite prefix of $Y$ is determined by a finite subset of $X$, the mapping $f$ is continuous. Furthermore, the function is surjective, because for every path $Y\!\in\!{\mathcal P}$ we can easily find some $X$ such that $f(X) = Y$. 

We now show that $f$ is measure preserving, i.e., that $\mu_t(f^{-1})(B) = \mu_w(B)$. Since Lebesgue measures are regular, it is sufficient to prove this property for basic clopen sets $B$.

Basic clopen sets are of the form $U_v\subseteq {\mathcal P}$, for some $v\in \{\dL, \dR\}^{*}$, consisting of all infinite paths extending the finite prefix $v$. We now show that $\mu_t(f^{-1}(U_v))= \mu_w(U_v)$ by induction of the length $| v| =n$ of the prefix $v$.

If $n=0$ then $v=\epsilon$ and $U_v=\mathcal{P}$ and therefore $f^{-1}(U_v)=\trees{\{0,1\}}$. Hence we have $\mu_t(f^{-1}(U_v))= \mu_w(U_v)=1$.

If $n>0$ assume that $v=w\dL$ (the case $v=w\dR$ is similar) and that the inductive hypothesis holds on $U_w$. By unfolding the definition of $f$ we get that the preimage $f^{-1}(U_v)$ is the clopen set
$$
f^{-1}(U_v) = f^{-1}(U_w) \cap U_{v=0}
$$
where $U_{v=0}= \{t \mid t(v)= 0 \}$. By definition of Lebesgue measure $\mu_t$ it holds that $\mu_t(U_{v=0})=\frac{1}{2}$. Furthermore, since $f^{-1}(U_w)$ and $U_{v=0}$ are independent events, it holds that
$$
\mu_{t}(f^{-1}(U_w) \cap U_{v=0}) = \mu_t(f^{-1}(U_w)) \cdot \frac{1}{2}
$$
By inductive hypothesis on $w$ we have that $\mu_t(f^{-1}(U_w)) = \mu_w(U_w)$. Lastly, since $v=w\dL$, we have $\mu_w(U_v) = \mu_w(U_w)\cdot \frac{1}{2}$. Therefore the desired equality
$$
\mu_t( f^{-1}(U_v)) = \mu_w(U_v)
$$
holds.
\end{proof}

We are now ready to prove the following theorem.
\begin{thm}
\label{encodabilty_of_quantifier}
For every $\StwoS+\qNpi$ formula $\phi(Z_1,\dots, Z_n)$ there exists a $\StwoS+\qN$ formula $\phi^\prime(Z_1,\dots, Z_n)$ such that $\phi$ and $\phi^\prime$ denote the same set.
\end{thm}

\begin{proof}
The proof goes by induction on the complexity of $\psi$ with the interesting case being $\phi(Z_1,\dots, Z_n)=\qNpi Y.\, \psi(Y,Z_1,\dots, Z_n)$. By induction hypothesis on $\psi$, there exists a  $\StwoS+\qN$ formula $\psi^\prime$ defining the same set as $\psi$. Then the $\StwoS+\qN$ formula $\phi^\prime$ corresponding to $\phi$ is defined as follows:
$$\phi^\prime(Z_1,\dots, Z_n) = \qN X.\, \big( \exists Y.\, \big( f(X, Y) \wedge \psi^\prime(Y,Z_1,\dots, Z_n) \big)\big)$$ 
We now show that $\phi$ and $\phi^\prime$ indeed define the same set. The following are equivalent:
\begin{enumerate}
\item The tuple $(t_1,\dots, t_n)$ satisfies the formula $\qNpi Y. \psi(Y,Z_1,\dots, Z_n)$,
\item (by the definition of $\qNpi X$) The set $A= \big\{ t \in \mathcal{P} \mid \psi(t,t_1,\dots, t_n) \big\}$ is such that $\mu_{w}(A)=1$,
\item (by Lemma~\ref{lemma_property_leadsto}) The set $B\subseteq \trees{\{0,1\}}$, defined as $B=f^{-1}(A)$, i.e., as  $$B = \big\{X\in \trees{\{0,1\}} \mid \exists Y.\,\big( f(X,Y) \wedge \psi(Y,t_1,\dots, t_n)\big)\big\}$$
is such that $\mu_{\trees{\{0,1\}}}(B)=1$.
\item (by the definition of $\qN$ and assumption $\psi=\psi^\prime$) the tuple $(t_1,\dots, t_n)$ satisfies 
    \begin{equation}
    \tag*{\qEd}
    \qN X.\, \big( \exists Y.\, \big( f(X,Y)  \wedge \phi^\prime(Y, \vv{Z}) \big)\big).
    \end{equation}
\end{enumerate}
\def\popQED{}
\end{proof}

The result of the previous theorem can be simply stated as $\StwoS+\qNpi \subseteq \StwoS+\qN$. 
The next result states that $\StwoS+\qNpi$ is strictly more expressive that $\StwoS$.

\begin{thm}\label{proposition:1:msopi}
The strict inequality $\StwoS \subsetneq \StwoS+\qNpi$ holds.
\end{thm}

Before proving this result, as a preliminary step, we introduce the following $\StwoS+\qNpi$ definable language of trees.
\begin{defi}
Let $\U^{=1}\subseteq \trees{\{0,1\}} $ be the set of
$\{0,1\}$-labeled trees $t$ satisfying the formula  $\U^{=1}(X)$ defined by:
$$\U^{=1}(X) = \qNpi Y. ( Y \cap X \textnormal{ is finite}).$$
\end{defi}
In other words $\U^{=1}$ is the collection of subsets $X$ of the full binary tree such that the set of paths having only finitely many vertices in $X$ has Lebesgue measure  $1$ in $\mathcal{P}$.

The following proposition states that every  $\StwoS+\forall^{=1}_\pi$ formula is equivalent to a \StwoS formula which, additionally, can use the additional predicate  over trees $U^{=1}$.

\begin{prop}\label{mnpiuniversalset1}
The equality $\StwoS+\qNpi = \StwoS+\U^{=1}$ holds. That is, every $\StwoS+\qNpi$ formula $\phi$ is effectively equivalent to a $\StwoS+\U^{=1}$ formula.
\end{prop}
\begin{proof}
The proof goes by induction on the structure of $\phi(X_1,\dots, X_n)$. The only interesting case is with $\varphi$ of the form $\qNpi Y .\psi(Y,X_1,\dots, X_n)$.

We can assume, from the inductive hypothesis on $\psi$, that  $\psi(Y,X_1,\dots, X_n)$ is expressible  in  $\StwoS+ \U^{=1}$.

By regularity of measures on Polish spaces a set of paths $B\subseteq \{\dL,\dR\}^\w$ has measure $1$ if and only if there is an $\Fsigma$ set $F\subseteq B$ such that $F$ has measure $1$. A standard argument (cf. proof of Theorem~\ref{qmpi_is_definable}, see also~\cite{Rabin69}) shows, that $F\subseteq\{\dL,\dR\}^\w$ is a $F_\sigma$ set if and only if $F$ is of the form $[Y]$ for a set $Y\subseteq \{\dL,\dR\}^\ast$ with

\[[Y]\eqdef\{\pi\in\{\dL,\dR\}^\w\mid \text{there are only finitely many  $v\in \pi$ such that $v\in Y$}\}.\]

Therefore, $\forall^{=1}_{\pi}\pi.\, \psi(\vec{X})$ holds if and only if there exists a set of vertuces $Y$ such that:
\begin{itemize}
\item  the set of paths $[Y]$ has measure $1$, and
\item  $[Y] \subseteq \lang( \psi(\vec{X})$.
\end{itemize}
Notice that the first property is expressible as $\U^{=1}(Y)$ and the second is  easily \StwoS-definable.
\end{proof}

We now show that the language $\U^{=1}$ is not definable in $\StwoS$. The proof is obtained by adapting the argument of Theorem 21 in~\cite{CHS2014}.

\begin{prop}\label{mnpiuniversalset2}
The language $\U^{=1}$ is not $\StwoS$ definable.
\end{prop}
\begin{proof}

A key property of $\StwoS$ definable sets $A\subseteq\trees{\{0,1\}}$ is that $A$ is not empty if and only if $A$ contains a regular tree, that is a tree having only infinitely many subtrees up-to isomorphism.

To prove that $\U^{=1}$ is not $\StwoS$ definable we will specify a $\StwoS$ definable language $L$ and show that $\U^{=1}\cap L$ is not empty and does not contain any regular tree. This of course implies that $\U^{=1}$ is not $\StwoS$ definable because $\StwoS$ definable sets are closed under Boolean operations.

We define $L\subseteq \trees{\{0,1\}}$ as the set of trees over the alphabet $\Sigma=\{0,1\}$ satisfying the $\StwoS$ formula $\phi(X)$ defined as:

$$\phi(X) =\forall x. \exists y . \big(x\leq y \wedge y\in X \big) \}$$

In other words, a set of vertices $t\in \trees{\{0,1\}}$ of the full binary tree satisfies $\phi(X)$ if from every vertex $v$ there exists a descendant vertex $w$ such that $w\in t$.

We now show that $\U^{=1}\cap L$ is not empty and does not contain a regular tree.

\begin{clm}
$\U^{=1}\cap L$ is not empty.
\end{clm}
\begin{proof}
We exhibit a concrete tree $t\!\in\trees{\{0,1\}}$ in $\U^{=1}\cap L$. To do this, fix any mapping $f:\mathbb{N}\to\mathbb{N}$ such that $f(0)\!=\!0$ and for all $n>0$ holds $f(n) > n + \sum^{n-1}_{i=0} f(i) $. 
We say that a vertex $v\!\in\!\{0,1\}^*$ of the full binary tree belongs to the $n$-th block if its depth $|v|$ is such that $f(n)\leq |v| < f(n+1)$. Each block can be seen as a forest of finite trees (see Figure~\ref{fig:l3_witness}) of depth $f(n+1)-f(n)$.
We now describe the tree $t$. For each $n$, all nodes of the $n$-th block are labeled by $0$ except the leftmost vertices of each (finite) tree in the block (seen as a forest), which are labeled by $1$. Figure~\ref{fig:l3_witness} illustrates this idea. Clearly $t$ is in $L$.


Let $E_{n}$ be the random event (on the space $\mathcal{P}$ of infinite branches of the full binary tree) of a path having the $f(n+1)$-th vertex labeled by $1$. Then, by construction of $t$, the (Lebesgue measure $\mu_w$) probability of $E_{n}$ is exactly $\frac{1}{2^{f(n+1) - f(n)}}$.

 
This implies that $\mu(E_0)+\mu(E_1)+\ldots \leq \sum_{n=0}^\infty \frac{1}{2^{f(n+1) - f(n)}}\leq \frac{1}{2} + \ldots + \frac{1}{2^n} \leq 1$. The Borel-Cantelli lemma implies that the probability of infinitely many events $E_n$ happening is $0$. Hence the probability of the set of paths having infinitely many $1$'s is $0$. Therefore $t\!\in\! \U^{=1}$ and thus $t\!\in\! L\cap\U^{=1}$.
\end{proof}


\def\prs{\tikz[scale=.65, every node/.style={scale=0.65}, baseline=1ex,shorten >=.1pt,node distance=1.8cm,on grid,semithick,auto,
every state/.style={fill=white,draw=black,circular drop shadow,inner sep=0mm,text=black},
accepting/.style ={fill=gray,text=white}]{
\node[state] (b) {$b$};
\begin{scope}[scale=.65, every node/.style={scale=0.65}, baseline=1ex,shorten >=.1pt,node distance=1.8cm, semithick,auto,
every state/.style={fill=white,draw=black,circular drop shadow,inner sep=0mm,text=black},
accepting/.style ={fill=gray,text=white}]
\node[state] (a1) [below left=1.7cm and 2.8cm  of b] {$1$};
\node[state] (b1) [below right=1.7cm and 2.8cm  of b] {$0$};
\node[state] (b2) [below left=1.7cm and 1.4cm  of b] {$0$};
\node[draw=none] (b3) [below right=1.7cm and 1.4cm  of b] {\ldots};
\node[draw=none] (b4) [below right=1.7cm and 0cm  of b] {\ldots};
\end{scope}
\draw [black,decorate,decoration=snake] (b) -- (a1);
\draw [black,decorate,decoration=snake] (b) -- (b1); 
\draw [black,decorate,decoration=snake] (b) -- (b2);
\draw [black,decorate,decoration=snake] (b) -- (b3);
\draw [black,decorate,decoration=snake] (b) -- (b4);
\sqsone; 
\sqstwo;
}
}

\def\sqsone{
\begin{scope}[scale=.55, every node/.style={scale=0.55}, baseline=1ex,shorten >=.1pt,node distance=1.8cm, semithick,auto,
every state/.style={fill=white,draw=black,circular drop shadow,inner sep=0mm,text=black},
accepting/.style ={fill=gray,text=white}]
\node[state] (a11) [below left=1.3cm and 2.2cm  of a1] {$1$};
\node[state] (b11) [below right=1.3cm and 2.2cm  of a1] {$0$};
\node[state] (b12) [below left=1.3cm and 1.1cm  of a1] {$0$};
\node[draw=none] (b13) [below right=1.3cm and 1.1cm  of a1] {\ldots};
\node[draw=none] (b14) [below right=1.3cm and 0cm  of a1] {\ldots};
\draw [black,decorate,decoration=snake] (a1) -- (a11);
\draw [black,decorate,decoration=snake] (a1) -- (b11); 
\draw [black,decorate,decoration=snake] (a1) -- (b12);
\draw [black,decorate,decoration=snake] (a1) -- (b13);
\draw [black,decorate,decoration=snake] (a1) -- (b14);
\end{scope}
}

\def\sqstwo{
\begin{scope}[scale=.55, every node/.style={scale=0.55}, baseline=1ex,shorten >=.1pt,node distance=1.8cm, semithick,auto,
every state/.style={fill=white,draw=black,circular drop shadow,inner sep=0mm,text=black},
accepting/.style ={fill=gray,text=white}]
\node[state] (a11) [below left=1.3cm and 2.2cm  of b1] {$1$};
\node[state] (b11) [below right=1.3cm and 2.2cm  of b1] {$0$};
\node[state] (b12) [below left=1.3cm and 1.1cm  of b1] {$0$};
\node[draw=none] (b13) [below right=1.3cm and 1.1cm  of b1] {\ldots};
\node[draw=none] (b14) [below right=1.3cm and 0cm  of b1] {\ldots};
\draw [black,decorate,decoration=snake] (b1) -- (a11);
\draw [black,decorate,decoration=snake] (b1) -- (b11); 
\draw [black,decorate,decoration=snake] (b1) -- (b12);
\draw [black,decorate,decoration=snake] (b1) -- (b13);
\draw [black,decorate,decoration=snake] (b1) -- (b14);
\end{scope}
}

\def\sqsthree{
\begin{scope}[scale=.35, every node/.style={scale=0.35}, baseline=1ex,shorten >=.1pt,node distance=1.8cm, semithick,auto,
every state/.style={fill=white,draw=black,circular drop shadow,inner sep=0mm,text=black},
accepting/.style ={fill=gray,text=white}]
\node[state] (q3) [below left=1.35cm and 0.4cm of q2r] {$q_1$};
\node[state] (s3) [below right=1.35cm and 0.4cm of q2r] {$q_2$};
\draw [black] (q2r) -- (q3);
\draw [black] (q2r) -- (s3); 
\end{scope}
}

\def\sqsfour{
\begin{scope}[scale=.35, every node/.style={scale=0.35}, baseline=1ex,shorten >=.1pt,node distance=1.8cm, semithick,auto,
every state/.style={fill=white,draw=black,circular drop shadow,inner sep=0mm,text=black},
accepting/.style ={fill=gray,text=white}]
\node[state] (q3) [below left=1.35cm and 0.4cm of p2] {$q_1$};
\node[state] (s3) [below right=1.35cm and 0.4cm of p2] {$q_3$};
\draw [black] (p2) -- (q3);
\draw [black] (p2) -- (s3); 
\end{scope}
}

\def\sqsfive{
\begin{scope}[scale=.35, every node/.style={scale=0.35}, baseline=1ex,shorten >=.1pt,node distance=1.8cm, semithick,auto,
every state/.style={fill=white,draw=black,circular drop shadow,inner sep=0mm,text=black},
accepting/.style ={fill=gray,text=white}]
\node[state] (q3) [below left=1.35cm and 0.4cm of q2] {$q_1$};
\node[state] (s3) [below right=1.35cm and 0.4cm of q2] {$q_1$};
\draw [black] (q2) -- (q3);
\draw [black] (q2) -- (s3); 
\end{scope}
}

\def\sqssix{
\begin{scope}[scale=.15, every node/.style={scale=0.15}, baseline=1ex,shorten >=.1pt,node distance=1.8cm, semithick,auto,
every state/.style={fill=white,draw=black,circular drop shadow,inner sep=0mm,text=black},
accepting/.style ={fill=gray,text=white}]
\node[state] (q5) [below left=1cm and 0.4cm of p3] {$p$};
\node[state] (s5) [below right=1cm and 0.4cm of p3] {$p$};
\draw [black] (p3) -- (q5);
\draw [black] (p3) -- (s5); 
\end{scope}
\node[draw=none] (s6) [below left=0.4cm and 0.35cm of s5,draw = none] {\ldots};
}

\def\sqsseven{
\begin{scope}[scale=.08, every node/.style={scale=0.08}, baseline=1ex,shorten >=.1pt,node distance=1.8cm, semithick,auto,
every state/.style={fill=white,draw=black,circular drop shadow,inner sep=0mm,text=black},
accepting/.style ={fill=gray,text=white}]
\node[state] (q6) [below left=0.7cm and 0.3cm of s5] {$q$};
\draw [black] (s5) -- (q6);
\end{scope}
\node[draw=none] (s6) [below right=0.7cm and 0.3cm of s5,draw = none] {\ldots};
\draw [black] (s5) -- (s6); 
}

\begin{figure}[H]
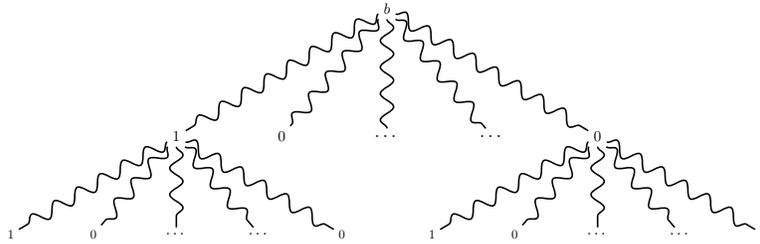

\centering
\prs

\caption{A prefix of a tree $t\in L\cap\U^{=1}$ up to the level $f(2)$.} 
\label{fig:l3_witness}
\end{figure}

\begin{clm}
$L\cap \U^{=1}$ does not contain any regular tree $t$. 
\end{clm}

\begin{proof}
Indeed, let $G$ be the finite graph whose vertices are labeled by $0$ or $1$ and where each vertex can reach exactly two vertices. Then $t$ represents a regular tree $t\in\trees{\{0,1\}}$. We can view $G$ as a finite Markov chain where all edges have probability $\frac{1}{2}$. From the assumption that $t\!\in\!L$, we know that every vertex in $G$ can reach a vertex labeled $1$. By elementary results of Markov chains, a random infinite path in $G$ will almost surely visit infinitely many times states labeled by $1$ and this is a contradiction with the hypothesis that $t\!\in\! \U^{=1}$.
\end{proof}

\noindent
The proofs of the above two Claims finish the proof of the Proposition.
\end{proof}

The results of Proposition~\ref{mnpiuniversalset1} and Proposition~\ref{mnpiuniversalset1} together prove the claim of Theorem~\ref{proposition:1:msopi}.

\paragraph*{\bf On Qualitative Automata of Carayol, Haddad, and Serre.} In a recent paper~\cite{CHS2014} Carayol, Haddad, and Serre have considered 
a probabilistic interpretation of standard (Rabin) nondeterministic tree automata. Below we briefly discuss this interpretation referring to~\cite{CHS2014} (see also~\cite[Chapter 8]{serre2015}) for more details. 

The classical interpretation from~\cite{Rabin69} of a nondeterministic tree automaton $\mathcal{A}$ over the alphabet $\Sigma$ is the set $\lang(A)\!\subseteq\! \trees{\Sigma}$ of trees $t\!\in \!\trees{\Sigma}$ such that there exists a \emph{run} $\rho$ of $t$ on $\mathcal{A}$ such that for \emph{all} paths $\pi$ in $\rho$, the path $\pi$ is accepting.
The probabilistic interpretation in~\cite{CHS2014}  associates to each nondeterministic tree automaton the language $\lang^{=1}(\mathcal{A})\!\subseteq\!\trees{\Sigma}$ of trees $t\!\in \!\trees{\Sigma}$ such that there exists a \emph{run} $\rho$ of $X$ on $\mathcal{A}$ such that for \emph{almost all} paths $\pi$ in $\rho$, the path $\pi$ is accepting, where ``almost all'' means having Lebesgue measure ($\mu_w$ relative to the space $\mathcal{P}$ of paths) equal to $1$ .

It is clear that for every automaton $\mathcal{A}$ the language $\lang^{=1}(\mathcal{A})$ can be defined in the logic $\StwoS+\qNpi$. Specifically, if $\Sigma=\{1,\dots, 2^n\}$, the  formula $\psi_{\mathcal{A}}(X_1,\dots, X_n)$ defined as
\[\psi_{A}(\vv{X}) =  \exists \vv{Y}. \big(\textnormal{``$\vv{Y}$ is a  run of $\vv{X}$ on $\mathcal{A}$''}\wedge\ \qNpi Z. (   \textnormal{``$Z$ is an accepting path of $\vv{Y}$''}) \big) \]
defines $\lang^{=1}(\mathcal{A})$, where the subformulas $\textnormal{``$\vv{Y}$ is a  run of $\vv{X}$ on $\mathcal{A}$''}$ and $\textnormal{``$\vv{Y}$ is a  run of $\vv{X}$ on $\mathcal{A}$''}$ are defined as expected (see, e.g.,~\cite{thomas96}).


On the other hand, Carayol, Haddad and Serre have shown in~\cite[Example 7]{CHS2014} that there are regular (i.e., $\StwoS$ definable) sets of trees which are not of the form $\lang^{=1}(\mathcal{A})$. Therefore there are $\StwoS+\qNpi$ languages not definable by automata with probabilistic interpretation.

As we have shown in Proposition~\ref{mnpiuniversalset1}, the (effective) equality $\StwoS+\qNpi=\StwoS+\U^{=1}$ holds. Interestingly, the (not regular) language $\U^{=1}$ is definable by tree automata with probabilistic interpretation of~\cite{CHS2014}.
Hence, $\StwoS+\qNpi$ can be understood as the minimal extension of $\StwoS$ which is sufficiently expressive to define all the languages definable by tree automata with probabilistic interpretation.

A crucial property of languages definable by tree automata with probabilistic interpretation is that if $\mathcal{L}^{=1}(\mathcal{A})\neq \emptyset$ then $\mathcal{L}^{=1}(\mathcal{A})$ contains a regular tree.
As we have shown in the proof of Theorem~\ref{mnpiuniversalset2} (claim 2), this useful property does not hold for $\StwoS+\qNpi$ definable languages.


\section{Open Problems}
\label{sec:problems}
\label{sec:open}

In this concluding section we present a list of problems left open from this work. This is by no means exhaustive and simply reflects the authors' view on what might be the most important, difficult and far-reaching questions encountered in the process of working on these topics.


\medskip
\noindent
\textbf{Problem 1.} Is the theory of $\StwoSM$ decidable?\\ 
The stronger statement $\StwoS=\StwoSM$ was formulated as a theorem in~\cite{MM2015} but, as discussed in the Introduction and Section~\ref{sec:s2s-category}), the proof contained a mistake. The proof is correct only for languages defined by game automata.

\medskip
\noindent
\textbf{Problem 2.} It is well known that $\StwoS$ definable sets always belong to the $\adelta{2}$ class of the projective hierarchy. Is this also true for $\StwoSM$ definable sets?

\medskip
\noindent
\textbf{Problem 3.} In terms of reverse mathematics as shown in~\cite{kolo_lics2016}, decidability of the theory $\StwoS$ is equivalent to determinacy of all Gale--Stewart games with the winning condition being a Boolean combination of $F_\sigma$ sets (see~\cite{kolo_lics2016} for a precise formulation). Is this characterization also valid for $\StwoSM$?






\medskip
\noindent
{\bf Problem 4.}  Is the theory  $\StwoS+\qNpi$ decidable?

\medskip
\noindent
{\bf Problem 5.}  Are  $\StwoS+\qNpi$ definable sets always contained in the $\adelta{2}$ class of the projective hierarchy. All the examples of $\StwoS+\qNpi$ definable sets which are not $\StwoS$ definable we have considered (e.g., the language $\U^{=1}$) are contained in the $\adelta{2}$ class.

\medskip
\noindent
{\bf Problem 6.}  Design an algorithm which for a given $\StwoS$ definable set computes its measure. 


\bibliographystyle{abbrv}
\bibliography{biblio}

\end{document}